\documentclass[11pt]{amsart}
\usepackage{amssymb,filecontents,enumerate}
\usepackage{tikz}
 
 \newcounter{counter:save}

\newtheorem{theorem}{Theorem}
\newtheorem{lemma}[theorem]{Lemma}
\newtheorem{corollary}[theorem]{Corollary}
\newtheorem{observation}[theorem]{Observation}
\newtheorem{remark}[theorem]{Remark} 
\newtheorem{fact}[theorem]{Fact}

 \def\calM{\mathcal{M}}

\def\Ztilde{\widetilde{Z}} 
\def\ZtildeTutte{\Ztilde}
\def\bgamma{\boldsymbol\gamma}
\def\bM{\boldsymbol M}
\def\tutte{\textsc{Tutte}}
\def\mtutte{\textsc{MatroidTutte}}
\def\multitutte{\textsc{MultiTutte}}
\def\signtutte{\textsc{SignTutte}}
\def\signtuttep{\signtutte(q,\gamma)}
\def\tuttep{\tutte(q,\gamma)}
\def\msigntutte{\textsc{MatroidSignTutte}}
\def\numP{\mathrm{\#P}} 
\def\FP{\mathrm{FP}}
\def\NP{\mathrm{NP}}
\def\numPQ{\numP_{\mathbb{Q}}}
\def\MCCut{{\sc \#Minimum Cardinality $(s,t)$-Cut}}
\def\hatw{\widehat{w}}
\def\hatbgamma{\widehat{\bgamma}}
\def\matrix{M}
\def\field{F}
\def\rows{V}
\def\columns{E}
\def\matroid{\calM}
\def\subsetcols{A}
\def\gf2{\mathbb{F}_2}
\def\rank{r}
\def\graph{G}
\def\graphvertices{V}
\def\graphedges{E}
\def\boldgamma{\bgamma}
\def\column{e}
\def\contract{/}

\let\epsilon=\varepsilon
\let\diamond=\Diamond
 
\def\prob#1#2#3{\begin{description}\item[\it Name]#1\item[\it Instance]#2\item[\it Output]#3\end{description}}

\title [The Complexity of Computing the Sign of the Tutte Polynomial]{The Complexity of Computing the Sign of the Tutte Polynomial}
 \thanks{This work was partially supported by the EPSRC grant EP/I011935/1
{\it Computational Counting}.
The research leading to these results has received funding from the European Research Council under 
the European Union's Seventh Framework Programme (FP7/2007-2013) ERC grant agreement no.\ 334828. The paper 
reflects only the authors' views and not the views of the ERC or the European Commission. 
The European Union is not liable for any use that may be made of the information contained therein.
A preliminary version
of these results was announced in the proceedings of
  ICALP 2012.}

\author{Leslie Ann Goldberg}
\address{Leslie Ann Goldberg, Department of Computer Science,
University of Oxford, Wolfson Building, Parks Road,
Oxford OX1~3QD, United Kingdom}

\author{Mark Jerrum}
\address{Mark Jerrum, School of Mathematical Sciences\\
Queen Mary, University of London, Mile End Road, 
London E1~4NS, United Kingdom.}

\begin{document}

\begin{abstract} 
We study the complexity of computing the sign of the Tutte polynomial of a graph.
As there are only three possible outcomes (positive, negative, and zero),
this seems at first sight more like a decision problem than a counting problem.
Surprisingly, however, there are large regions of the parameter space 
for which computing the sign of the Tutte polynomial
is actually \#P-hard. As a trivial consequence, approximating the polynomial is 
also \#P-hard in this case.
Thus, {\it approximately\/} evaluating the Tutte polynomial 
in these regions is as hard as {\it exactly\/} counting the satisfying
assignments to a CNF Boolean formula.
For most  other points in the parameter space, 
we show that computing the sign of the polynomial is in FP, whereas
approximating the polynomial can be done in polynomial time with an NP oracle.
As a special case, we completely resolve the complexity of computing the sign of the chromatic polynomial ---
this is easily computable at $q=2$ and when $q\leq 32/27$, and is NP-hard to compute for
all other values of the parameter~$q$.
\end{abstract}

\maketitle

\section{Introduction} 
The Tutte polynomial of an undirected\footnote{All graphs in this article
are undirected, so we shall drop the qualifier in what follows.} graph
is a two-variable polynomial that
captures many interesting properties of the graph 
such as (by making appropriate choices of the two variables)
the number of $q$-colourings, the number of nowhere-zero $q$-flows,
the number of acyclic orientations, and the probability that the graph 
remains connected when edges are deleted at random.

Much work~\cite{ETHb,ferropotts,tutteone,planartutte,JVW90,vertigan}
has studied the difficulty of evaluating the polynomial
(exactly or approximately) when the values of the variables
are fixed, and a graph is given as input.

Our early paper~\cite{tutteone} identified a large region of points where
the approximate evaluation of the polynomial is NP-hard,
and a short hyperbola segment along which approximate evaluation 
is even \#P-hard.
Thus, an approximation of the polynomial at a point on this short hyperbola segment
would enable one to \emph{exactly} solve a problem in \#P\null.
Kuperberg~\cite[Theorem 1.3] {kuperberg} uses quantum results
to show  similar (classical) \#P-hardness  for all points $(x,y)$ in the negative quadrant
satisfying $(x-1)(y-1)>4$.
In this paper, we show that, in fact, for most of the NP-hard points identified in \cite{tutteone},
approximation is \#P-hard. 
Moreover, it
is \#P-hard for a very simple reason:  determining the sign of the
polynomial --- i.e., whether the evaluation of the polynomial is positive, negative or
zero --- is \#P-hard.
This seems surprising since determining the sign of the polynomial is
nearly a decision problem (there are only three possible outcomes) but
it is \#P-hard nearly everwhere (at all of the red points in the plane in Figure~\ref{fig:one}).

Past work~\cite{JacksonSokal} has studied the sign of the Tutte polynomial ---  in particular, 
Jackson and Sokal sought to determine for which
choices of the two variables the sign is ``trivial'' in the sense that
it does not depend on the input graph (or it depends only
very weakly on the input graph, for example when it depends only on
the number of vertices in the graph).

To illustrate how our work fits in with the work of Jackson and Sokal,
we start with an important univariate case.
The \emph{chromatic polynomial} $P(G;q)$ of an $n$-vertex graph~$G$ is
the unique degree-$n$ polynomial in the variable~$q$ such
that $P(G;q)$ is the number of proper $q$-colourings of~$G$.
Jackson \cite[Theorem 5]{Jackson}
showed that for $q\in(1,32/27]$ the sign of $P(G;q)$ depends upon $G$ in an essentially trivial way.
In particular, for every connected simple graph with $n\geq 2$ vertices and $b$~blocks,
$P(G;q)$ is non-zero with sign ${(-1)}^{n+b-1}$.
The sign of $P(G;q)$ is also known to be a trivial function of~$G$ for $q\leq 1$.
(See, for example, \cite[Theorem 1.1]{JacksonSokal}.)
Jackson~\cite[Theorem 12]{Jackson} 
demonstrated the significance of the value~$32/27$ by constructing an infinite family of graphs 
such that $P(G;q)=0$ at a value of~$q$ which is  arbitrarily close to~$32/27$.
In fact, Jackson and Sokal conjectured \cite[Conjecture 10.3(e)]{JacksonSokal} 
that the value~$32/27$ is a phase transition in the sense
that, for every~$q$ above this critical value, the sign of~$P(G;q)$ is a non-trivial function of~$G$. 
In particular, they conjectured that for any
fixed $q>32/27$, and all sufficiently large~$n$ and~$m$, there 
are $2$-connected graphs~$G$ with $n$ vertices and $m$ edges
that make $P(G;q)$ non-zero with either sign.

It turns out that this intuition is correct
(see Corollary~\ref{cor:meetJS})
and 
that $q=32/27$ is, in some sense, a phase transition for the complexity of computing the
sign of $P(G;q)$: 
\begin{itemize}
\item As was known, for $q\leq 32/27$, the sign of $P(G;q)$ is a trivial function of~$G$, 
which is easily computed.
\item At $q=2$, the evaluation $P(G;q)$ is the number of $2$-colourings of~$G$.
The sign of~$P(G;q)$ is positive if $G$ is bipartite, and is~$0$ otherwise.
Thus, the sign of $P(G;q)$ is not a trivial function of~$G$, 
but $P(G;q)$ is still easily computed in polynomial time.
\item For every $q>32/27$ except $q=2$,  
computing the sign of $P(G;q)$ is NP-hard. 
\end{itemize}
However, the full version of Jackson and Sokal's conjecture turns out  to be incorrect. 
See Observations~\ref{obs:Fpos} and~\ref{obs:Epos} for counter-examples.

While computing the sign of $P(G;q)$ is NP-hard for every $q\neq 2$ which is greater than~$32/27$,
the precise complexity of computing the sign does actually depend upon~$q$.
We show (see Corollary~\ref{cor:meetJS}) that
for each fixed {\it non-integer} $q>32/27$, the complexity of computing the sign of $P(G;q)$ is
\#P-hard.
This means that a 
polynomial-time algorithm for computing the sign of $P(G;q)$, given~$G$,
would give a 
polynomial-time algorithm for 
exactly solving every problem in \#P\null.
On the other hand, for integers $q>2$, the problem of computing the sign of $P(G;q)$ 
is merely NP-complete.\footnote{As there are three potential outcomes, 
determining the sign cannot be NP-complete in a strict sense.  However, in this 
case, one of the outcomes (negative) is impossible, so we can view the determination of 
the sign as an NP-problem by identifying positive with ``accept'' and zero with
``reject''.  This view will be taken throughout the paper.}

As one would expect,
both of these results have ramifications for the complexity of approximating $P(G;q)$.
A fully polynomial approximation scheme (FPRAS) for evaluating $P(G;q)$, given~$G$,
can be used as a polynomial-time randomised algorithm for computing the sign of~$P(G;q)$.
Thus, we can immediately deduce that if $q$ is a non-integer greater than $32/27$,
then there is no FPRAS for $P(G;q)$ unless there is a randomised polynomial-time
algorithm for exactly solving every problem in~\#P.  See Section~\ref{sec:ranAlg}
for a more thorough discussion of this claim.
 
On the other hand, for integer values $q>32/27$, we show that the problem of evaluating 
$P(G;q)$
is in the complexity class $\numPQ$, which is defined as follows.

{\bf Definition.} 
$\FP$ is  the class of
functions computable by polynomial-time algorithms.
We say that a function 
$f:\Sigma^{*}\to\mathbb{Q}$ is in the class $\numPQ$ if
$f(x)=a(x)/b(x)$, where $a,b:\Sigma^{*}\to\mathbb{N}$, and
$a\in\numP$ and $b\in\mathrm{FP}$.  

If $f$ is in $\numPQ$ then there is an approximation scheme for~$f$ that runs in polynomial
time, using an oracle for an NP predicate (for a more detailed discussion, see~\cite[Section 2.2]{tutteone}).
Thus, it is presumably much easier to 
approximate
$P(G;q)$ when $q$ is an integer greater than~$32/27$,
as compared to a non-integer.

All of these considerations generalise smoothly to the Tutte polynomial,
which we now define.
Since we will later need the multivariate generalisation~\cite{sokal} of the polynomial,
we use the ``random cluster'' formulation of the
Tutte polynomial, which for a graph~$G=(V,E)$,
is defined as a polynomial in indeterminates~$q$ and~$\gamma$ 
as follows,
\begin{equation}
\label{eq:Zdef}
 Z(G;q,\gamma)=\sum_{A\subseteq E}
q^{\kappa(V,A)}  \gamma^{|A|},
\end{equation}
where $\kappa(V,A)$ denotes the number of connected components in the
graph~$(V,A)$.
The chromatic polynomial discussed earlier is related to the Tutte polynomial
via the identity \cite[(2.15)]{JacksonSokal}
$P(G;q) = Z(G;q,-1)$.

In fact, Tutte  defined the Tutte polynomial using a different, two-variable
parameterisation, in terms of variables~$x$ and~$y$.
This polynomial is defined for a graph $G=(V,E)$ by 
\begin{equation}
\label{eq:Tdef}
T(G;x,y) = \sum_{A \subseteq E} {(x-1)}^{\kappa(V,A)-\kappa(V,E)}{(y-1)}^{|A|-|V|+\kappa(V,A)}.
\end{equation}
It is well known (see, for example, \cite[(2.26)]{sokal})
that when $q=(x-1)(y-1)$ and
$\gamma=y-1$ we have
\begin{equation}
\label{eq:ZtoT}
T(G;x,y) = {(y-1)}^{-|V|}{(x-1)}^{-\kappa(V,E)} Z(G;q,\gamma).
\end{equation}

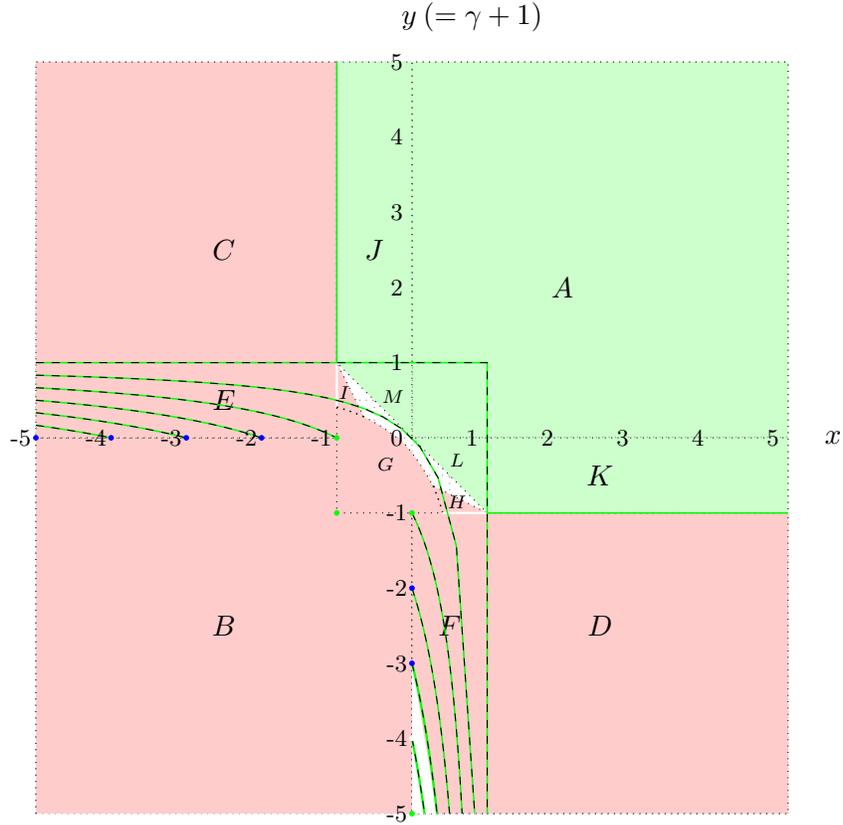
\begin{figure}[t]
\centering{
 \begin{tikzpicture}
 \draw[step=0.5cm,color=gray,style=dotted] (-5,-5) grid (5,5);
 \draw (0.8,5.6) node{$y\>(=\gamma+1)$};
 \draw(5.6,0) node{$x$};
 \draw [fill=green!20,style=dotted] (0,0) rectangle  (5,5); \draw(2,2) node {$A$};
 \draw [fill=green!20,style=dotted] (1,-1) rectangle  (5,0);\draw (2.5,-0.5) node {$K$};  
 \draw [fill=red!20,style=dotted] (1,-5) rectangle  (5,-1);\draw (2.5,-2.5) node {$D$};
 \draw [fill=red!20,style=dotted] (-5,1) rectangle  (-1,5);\draw (-2.5,2.5) node {$C$};
 \draw [fill=red!20,style=dotted] (-5,-5) -- (-5,0) -- (-1,0) %start in the region x<-1
-- (-1,1) -- % go up to pick up the triangle y<-1-2x in the vicinity of the origin, carry along the triangle until
% point (-0.64,0.28). At this point,
(-0.64,0.28) ..  controls (-0.46,0.19) and (-0.27,0.07).. (-0.09,-0.09)
--
(-0.09,-0.09) ..  controls (0.03,-0.23) and (0.15,-0.40).. (0.28,-0.64)
% pick up the hyperbola q=32/27 in the vicinity of the origin
% I don't know how to do it properly, with the fill, so instead I am representing it with
% two Bezier plots, each of which uses two endpoints and two controls
%% A quick way to generate the points is (in mathematica)
%% HEQ [q_] := (x - 1) (y - 1) == q;
%% Solveit[b_] := Solve[{HEQ[1.5], x == b}, y]
%% Table[Solveit[x], {x, { -0.5, -0.4, -0.1, 0}}]
%% NB we have another hyperbola q=1 below produced similarly
% After the hyperbola, carry on the triangle x<-1-2y in the vicinity of the origin
-- (1,-1) -- (0,-1)  % then carry on to finish y<-1 and go back to the origin, (-4,-4) to shade inside red.
--(0,-5) (-5,5); 
\draw (-2.5,-2.5) node {$B$};
 \draw [fill=red!20,style=dotted] (-5,0) rectangle  (-1,1);\draw (-2.5,0.5) node {$E$};
 \draw [fill=red!20,style=dotted] (0,-5) rectangle  (1,-1);\draw (0.5,-2.5) node {$F$}; % Note that I swapped
 % the names of E and F so that the easier one (now E) with x<-1 will come first so lemma names, etc., are all wrong
  \draw (0.6,-0.85) node {\tiny $H$};
 \draw (-0.9,0.6) node {\tiny $I$}; 
  \draw (-0.35,-0.35) node {\tiny $G$};  
  \draw [fill=green!20,style=dotted] (-1,1) rectangle  (0,5);\draw (-0.5,2.5) node {$J$}; 
\draw[fill=green!20,style=dotted] (0,0) -- (-1,1)--(0,1)--(0,0)--(1,0)--(1,-1)--(0,0);  
 % new regions 
 \draw (0.6,-0.3) node {\tiny $L$};
 \draw (-0.25,0.55) node {\tiny $M$};
       \draw[style=dotted] (-1,0)--(-1,-1)--(0,-1); % boundary between B and the vicinity of the origin, which I am now calling $G$
 \draw [color=green,fill=green](-1,-1) circle (0.3mm); % special point 
 \draw[color=white,line width=0.3mm](-1,1)--(-1,0.407); % The FI boundary is still open   
 \draw[color=white,line width=0.3mm](1,-1)--(0.407,-1); % The EH boundary is still open   
 \draw[color=green,line width=0.2mm] (-5,1)--(1,1)--(1,-5); % The Hyperbola q=0
 \draw[style=dashed] (-5,1)--(1,1)--(1,-5); % The Hyperbola q=0 for black-and-white viewers
  \draw[color=green,domain=-5:5/6, line width=0.2mm] plot (\x,{(1/(\x-1)+1)}); % The Hyperbola q=1  
 \draw[style=dashed,domain=-5:5/6] plot (\x,{(1/(\x-1)+1)}); % The Hyperbola q=1  for black-and-white
 \draw[color=green,domain=-5:-1, line width=0.2mm] plot (\x,{(2/(\x-1)+1)}); % The Hyperbola q=2 E
  \draw[style=dashed,domain=-5:-1] plot (\x,{(2/(\x-1)+1)}); % The Hyperbola q=2 black-and-white
 \draw[color=green,domain=-5:-2, line width=0.2mm] plot (\x,{(3/(\x-1)+1)}); % The Hyperbola q=3  E
  \draw[style=dashed,domain=-5:-2 ] plot (\x,{(3/(\x-1)+1)}); % The Hyperbola q=3  black-and-white
 \draw[color=green,domain=-5:-3, line width=0.2mm] plot (\x,{(4/(\x-1)+1)}); % The Hyperbola q=4  E
  \draw[style=dashed,domain=-5:-3 ] plot (\x,{(4/(\x-1)+1)}); % The Hyperbola q=4  black-and-white
 \draw[color=green,domain=-5:-4, line width=0.2mm] plot (\x,{(5/(\x-1)+1)}); % The Hyperbola q=5  E
  \draw[style=dashed,domain=-5:-4 ] plot (\x,{(5/(\x-1)+1)}); % The Hyperbola q=5  black-and-white
     \draw[color=green,domain=0:4/6, line width=0.2mm] plot (\x,{(2/(\x-1)+1)}); % The Hyperbola q=2 F
          \draw[ style=dashed,domain=0:4/6 ] plot (\x,{(2/(\x-1)+1)}); % The Hyperbola q=2 black-and-white
 \draw[color=green,domain=0:3/6, line width=0.2mm] plot (\x,{(3/(\x-1)+1)}); % The Hyperbola q=3  F
  \draw[style=dashed,domain=0:3/6 ] plot (\x,{(3/(\x-1)+1)}); % The Hyperbola q=3  black-and-white
\draw[color=black,style=dotted,domain=-1:-0.64, line width=0.2mm] plot (\x,{(1.18/(\x-1)+1)}); % The Hyperbola q=32/27
\draw[color=black,style=dotted,domain=0.28:0.407, line width=0.2mm] plot (\x,{(1.18/(\x-1)+1)}); % The Hyperbola q=32/27
%filling in boundary of I and H
%%Redo the bit of region F that I don't know!
\fill[white](0,-5)--(0,-3) -- (2/6,-5)--(0,-5); % This is  under q=5, which I am approximating by a straight line!
%(I don't quite know how to put a hyperbola there)
 \draw[color=green,domain=0:2/6, line width=0.3mm] plot (\x,{(4/(\x-1)+1)}); % The Hyperbola q=4  F
  \draw[style=dashed,domain=0:2/6 ] plot (\x,{(4/(\x-1)+1)}); % The Hyperbola q=4  black-and-white
  \draw[color=green,domain=0:1/6, line width=0.3mm] plot (\x,{(5/(\x-1)+1)}); % The Hyperbola q=5  F
    \draw[style=dashed,domain=0:1/6 ] plot (\x,{(5/(\x-1)+1)}); % The Hyperbola q=5  black-and-white
%These boundaries seem to be tractable
\draw[color=green,line width=0.2mm] (-1,1)--(-1,5);
\draw[color=green,line width=0.2mm] (1,-1)--(5,-1);
    \foreach \i in {-5,-4,-3,-2,-1,0,1,2,3,4,5}
  {\draw (0-0.2,\i) node{\footnotesize \i};}
  \foreach \i in {-5,-4,-3,-2,-1,1,2,3,4,5}
   {\draw (\i-0.2,0) node{\footnotesize \i};}  
  \draw [color=blue,fill=blue](-5,0) circle (0.3mm);    
   \draw [color=blue,fill=blue](-4,0) circle (0.3mm);  
    \draw [color=blue,fill=blue](-3,0) circle (0.3mm);  
      \draw [color=blue,fill=blue](-2,0) circle (0.3mm);  
        \draw [color=green,fill=green](-1,0) circle (0.3mm);  
 \draw [color=green,fill=green](0,-1) circle (0.3mm);  
 \draw [color=blue,fill=blue](0,-2) circle (0.3mm);  
  \draw [color=blue,fill=blue](0,-3) circle (0.3mm);  
  \draw [color=white,fill=white](0,-4) circle (0.3mm);  
   \draw [color=green,fill=green](0,-5) circle (0.3mm);  
  \end{tikzpicture}}
 \caption{An illustration of Theorem~\ref{thm:main}.
 Computing the sign of the Tutte polynomial is $\numP$-hard at red points, 
 is $\NP$-complete at blue points, and
 is in $\FP$ at green points.  We have not resolved the complexity at white points.
 At red points, approximating the Tutte polynomial is also $\numP$-hard. At blue and green points,  
 it can be done in polynomial time with an $\NP$~oracle.
 Guide for the greyscale version: The red points appear as a darker grey in regions B, C, D, E, F, G, H and~I.
 The green points appear as a lighter grey in regions A, J, K, L and M and also as dashed 
 hyperbola segments and at the points $(-1,0)$, $(-1,-1)$, $(0,-1)$ and $(0,-5)$.
 The blue points are $(-2,0)$, $(-3,0)$, $(-4,0)$, $(-5,0)$, $(0,-2)$ and $(0,-3)$.
 }
 \label{fig:one}
 \end{figure} 
 
This paper studies the complexity of computing the sign of the 
(random cluster) Tutte polynomial.  The definitive statement of our results requires a number of 
formal definitions and is  presented as Theorem~\ref{thm:main} in Section~\ref{sec:mainthm}. 
However, an informal description of Theorem~\ref{thm:main} appears in
Figure~\ref{fig:one}, which illustrates the   
the $(x,y)$ plane divided into a number
of regions A--M according to their complexity.\footnote{For convenience,
 our proofs use the random cluster formulation of the Tutte polynomial~(\ref{eq:Zdef}). However,
 in order to make our results easily comparable to other results in the literature such as~\cite{tutteone}
 and~\cite{JVW90}, we classify points using the $(x,y)$-coordinatisation of~(\ref{eq:Tdef}). This is without loss
 of generality, since it is easy to go from one coordinate system to the other using~(\ref{eq:ZtoT}).
 However, the reader should note that  
 if $y=1$ then 
 $\gamma=0$ and
 $q=(x-1)(y-1)=0$ so computing $Z(G;q,\gamma)$ is trivial,
 whereas the complexity of computing $T(G;x,y)$ is unclear.
 In general, any two-parameter version of the Tutte polynomial will omit some points. This issue is discussed
 further in \cite[Section 1]{binarymatroidtutte}.
    }
The colours depict the complexity of computing the sign of the polynomial
for a fixed point $(x,y)$. If the point $(x,y)$ is coloured red, then 
the problem of computing the sign is \#P-hard.
If the point $(x,y)$ is coloured green, then the problem of computing the sign is in FP\null.
Finally, if the point $(x,y)$ is coloured blue, then the  problem of computing the sign is
NP-complete. (There are still some points for which we have not resolved the complexity --- these
are coloured white.)

To resolve any ambiguities in Figure~\ref{fig:one}, 
a formal description of the regions appearing there
is provided in Figure~\ref{fig:regions}.
For each region of interest, 
the condition for a point $(x,y)$ to
belong to that region is given.
Note that $q$ is used to denote $(x-1)(y-1)$.

\begin{figure}[t]
\begin{itemize}
\item Region A: $x\geq 0$ and $y\geq 0$.  
\item Region B: 
$\min(x,y) \leq -1$ and $\max(x,y)<0$. 
\item Region C:  $x<-1$ and $y>1$. 
\item Region D:  $x>1$ and $y<-1$. 
\item Region E: $x\leq -1$ and $0< y \leq 1$. 
\item Region F: $0< x \leq 1$ and $y\leq -1$. 
\item The boundary between regions~B and~E: $x\leq -1$ and $y=0$.
\item The boundary between regions~B and~F:  $x=0$ and $y\leq -1$.
\item Region G:  $\max(|x|,|y|)<1$ and $q>32/27$.
\item Region H: $\max(|x|,|y|)<1$ and $q\leq 32/27$ and $x<-2y-1$. 
\item Region I: $\max(|x|,|y|)<1$ and $q\leq 32/27$ and $y<-2x-1$.
\item Region J: $-1\leq x < 0$ and $y\geq 1$. 
\item Region K:  $x\geq 1$ and $-1\leq y < 0$. 
\item Region L: $0<x<1$ and $-x<y<0$.
\item Region M: $0<y<1$ and $-y<x<0$.
\item The rest: There are some remaining unresolved points. These points (simultaneously) satisfy
all of the following inequalities:  $\max(|x|,|y|)<1$, $y<-x$, 
$q\leq 32/27$, $y\geq -2x-1$, $x\geq -2y-1$,   and  
$q\neq 1$.
\end{itemize}
\caption{A formal description of the regions appearing in 
Theorem~\ref{thm:main} and Figure~\ref{fig:one}.
For each region of interest, we give the condition for a point $(x,y)$ to
belong to that region.  Throughout we use $q$ to denote $(x-1)(y-1)$.}
\label{fig:regions}
\end{figure} 

Once again, there are ramifications for the complexity of approximating the Tutte polynomial.
Since an FPRAS for $Z(G;q,\gamma)$ gives a randomised algorithm for computing its sign,
we can again deduce that there is no FPRAS for points that are coloured red (unless  
there is a randomised polynomial-time algorithm for exactly solving every problem in \#P).
By contrast, for all of the points that are coloured green or blue, we also show that
the problem of computing $Z(G;q,\gamma)$ is in the complexity class $\numPQ$.
Thus, the polynomial can be approximated in polynomial-time using an NP oracle.
 
In order to reach into some of the regions, for example~F, it has been
necessary to use gadgets that go beyond the series-parallel graphs that 
have so-far proved adequate in this area.  For example, exploring region~F has 
necessitated the use of a gadget based on the Petersen graph.   

Our classification is not complete and leaves unresolved some areas
(coloured white in Figure~\ref{fig:one}).  Although the methods could no doubt
be pushed a little further, at the expense of adding further complexity
to the proofs, it seems likely that a complete classification is some
way off.  For example, showing that sign of the Tutte polynomial is
hard to compute at the point $(0,-4)$ would necessarily provide a
counterexample to Tutte's long-standing 5-flow conjecture.  In the other
direction, it is difficult to conceive of an efficient algorithm for
deciding the sign that would not at the same time resolve the
conjecture.

\section{Preliminaries}
\label{sec:preliminaries}

\subsection{The Tutte polynomial}

It will   be helpful to define the multivariate version of the 
random cluster formulation of the Tutte polynomial.
Let $\bgamma$ be a function that assigns a (rational) weight $\gamma_e$
to every edge~$e\in E$. We refer to $\bgamma$ as a ``weight function''.
We define 
$$ Z(G;q,\bgamma)=\sum_{A\subseteq E} 
q^{\kappa(V,A)}
\prod_{e\in A} \gamma_e.
$$

Given a graph~$G=(V,E)$ with distinguished nodes~$s$ and~$t$,
$Z_{st}(G;q,\bgamma)$  denotes the contribution to
$Z(G;q,\bgamma)$ arising from edge-sets~$A$ in which~$s$ and~$t$
are in the same component of $(V,A)$.
That is, 
$$Z_{st}(G;q,\bgamma) = \sum_{A\subseteq E: \text{$s$ and $t$ in same component}} 
q^{\kappa(V,A)}
\prod_{e\in A} \gamma_e.
$$
Similarly, $Z_{s|t}$ denotes the contribution
arising from edge-sets~$A$ in which $s$ and $t$ are in different components, so
$Z(G;q,\bgamma) = Z_{st}(G;q,\bgamma) + Z_{s|t}(G;q,\bgamma)$.

\subsection{Implementing new edge weights, series compositions and parallel compositions }
\label{sec:shiftdef}

Our treatment of implementations, series compositions and parallel compositions
is completely standard and is taken from \cite[Section 2.1]{planartutte}. The reader
who is familiar with this material can skip this section (which is included here for completeness).

Let 
$W$ be a set of (rational)
edge weights  
and fix a value $q$.
Let $w^*$ be a weight (which may not be in $W$) which we want to ``implement''.
Suppose that 
there is a graph~$\Upsilon$,
with distinguished vertices $s$ and~$t$ 
and a weight function $\hatbgamma: E(\Upsilon) \rightarrow W$
such that
\begin{equation}
\label{eq:implement}
w^* = q Z_{st}(\Upsilon;q,\hatbgamma)/Z_{s|t}(\Upsilon;q,\hatbgamma).
\end{equation}
In this case, we say that $\Upsilon$ and $\hatbgamma$ implement $w^*$
(or even that $W$ implements~$w^*$).

The purpose of ``implementing''  edge weights is this.
Let $G$ be a graph with weight function $\bgamma$.
Let $f$ be some edge of~$G$ with weight $\gamma_f=w^*$.
Suppose that $W$ implements $w^*$.
Let $\Upsilon$ be a graph with distinguished vertices $s$ and $t$
with a weight function $\hatbgamma:E(\Upsilon)\rightarrow W$ 
satisfying~(\ref{eq:implement}). 
Construct the weighted graph~$G'$ 
by replacing edge $f$ with a copy of~$\Upsilon$ (identify $s$ with either endpoint of~$f$
(it doesn't matter which one) and identify $t$ with the other endpoint of $f$ and remove edge $f$).
Let the weight function 
$\bgamma'$ of $G'$ inherit weights from $\bgamma$ and 
$\hatbgamma$ (so $\gamma'_e=\hat\gamma_e$ if $e\in E(\Upsilon)$ and
$\gamma'_e = \gamma_e$ otherwise).
Then the definition of the multivariate Tutte polynomial gives
\begin{equation}
\label{eq:shift}
Z(G';q,\bgamma') = \frac{Z_{s|t}(\Upsilon;q,\hatbgamma)}{q^2} Z(G;q,\bgamma).\end{equation}
So, as long as $q\neq 0$ and $Z_{s|t}(\Upsilon;q,\hatbgamma)$ is easy to evaluate,
evaluating the multivariate Tutte polynomial of $G'$ with weight function $\bgamma'$ is
essentially the same as evaluating the multivariate Tutte polynomial of $G$ with weight function~$\bgamma$.
 
Two especially useful implementations are series and parallel compositions.
These are explained in detail in \cite[Section 2.3]{JacksonSokal}.
So we will be brief here.
Parallel composition is the case in which $\Upsilon$ consists of two parallel edges $e_1$ and $e_2$
with endpoints $s$ and $t$ and  $\hat\gamma_{e_1}=w_1$ and $\hat{\gamma}_{e_2}=w_2$.
It is easily checked from Equation~(\ref{eq:implement})
that $w^* = (1+w_1)(1+w_2)-1$. Also, the extra factor in Equation~(\ref{eq:shift}) cancels,
so in this case $Z(G';q,\bgamma') = Z(G;q,\bgamma)$.

Series composition is the case in which $\Upsilon$ is a length-2 path from $s$ to $t$ consisting of edges $e_1$ and $e_2$
with $\hat\gamma_{e_1}=w_1$ and $\hat\gamma_{e_2}=w_2$.
It is easily checked from Equation~(\ref{eq:implement})
that $w^* =  w_1w_2/(q+w_1+w_2)$. Also, the extra factor in Equation~(\ref{eq:shift}) is $q+w_1+w_2$,
so in this case $Z(G';q,\bgamma') = (q+w_1+w_2) Z(G;q,\bgamma)$.
It is helpful to note that
$w^*$ satisfies
$$\left(1+\frac{q}{w^*}\right) = \left(1+\frac{q}{w_1}\right) \left(1+\frac{q}{w_2}\right).$$

We say that there is a ``shift'' 
from $(q,\alpha)$ to $(q,\alpha')$ if 
there is an implementation of $\alpha'$ consisting of some $\Upsilon$ and  
$\hatw:E(\Upsilon)\rightarrow W$ where $W$ is the singleton set $W=\{\alpha\}$.
This is the same notion of ``shift'' that we used in \cite{tutteone}. 
Taking $y=\alpha+1$ and $y'=\alpha'+1$
and defining $x$ and $x'$ by $q=(x-1)(y-1)=(x'-1)(y'-1)$ we
equivalently refer to this as a shift from $(x,y)$ to $(x',y')$.
It is an easy, but important observation that shifts may be composed to 
obtain new shifts.  So, if we have shifts from $(x,y)$ to $(x',y')$
and from $(x',y')$ to $(x'',y'')$, then we also have a shift 
from $(x,y)$ to $(x'',y'')$.

The
$k$-thickening of  \cite{JVW90}
is the parallel composition of $k$ edges of weight $\alpha$.
It implements $\alpha'=(1+\alpha)^k-1$ and is a shift from $(x,y)$ to
$(x',y')$ where $y' = y^k$ (and $x'$ is given by $(x'-1)(y'-1)=q$).
 Similarly, the $k$-stretch is the series composition of $k$ edges of weight $\alpha$.
It implements an $\alpha'$
satisfying
$$1+\frac{q}{\alpha'}= {\left(1+\frac{q}{\alpha}\right)}^k,$$
It is a shift from $(x,y)$ to 
$(x',y')$ where $x'=x^k$.
(In the classical bivariate $(x,y)$ parameterisation, there is effectively
one edge weight, so the stretching or thickening is applied uniformly
to every edge of the graph.)

Since it is useful to switch freely between $(q,\alpha)$ coordinates and $(x,y)$ coordinates we also
refer to the implementation in Equation~(\ref{eq:implement}) as an implementation of the 
point $(x,y)=(q/w^*+1,w^*+1)$ using the points
$$\{
(x,y)=(q/w+1,w+1)\mid w\in W
\}.$$
Thus, 
if $q=(x_1-1)(y_1-1)=(x_2-1)(y_2-1)$ then
the series composition of $(x_1,y_1)$ and $(x_2,y_2)$ implements
the point 
$$\left(\frac{q}{y_1 y_2-1}+1,y_1 y_2\right),$$
and the parallel composition of these implements the point
$$\left(x_1 x_2,\frac{q}{x_1 x_2-1}+1\right).$$
We make extensive use of series and
parallel composition, and the above identities will be employed without 
comment. 

\subsection{Computational Problems}
\label{sec:compprobs}
For fixed rational numbers~$q$, $\gamma$ and $\gamma_1,\ldots,\gamma_k$, we
consider the following computational problems\footnote{In \cite{tutteone} we
referred to these as $\multitutte(q,\gamma)$ and $\multitutte(q;\gamma_1,\ldots,\gamma_k)$ respectively,
but we use the shorter names here since there is no confusion.} from~\cite{tutteone}.

\prob{$\tutte(q,\gamma)$.}
{A graph $G=(V,E)$.}
{The rational number $Z(G;q,\gamma)$.} 
\prob{$\tutte(q;\gamma_1,\ldots,\gamma_k)$.}
{A graph $G=(V,E)$ and
a weight function $\bgamma:E\rightarrow \{\gamma_{1},\ldots,\gamma_{k}\}$.}
{The rational number $Z(G;q,\bgamma)$.}
We also consider variants in which the goal is to compute the sign of the Tutte polynomial.
\prob{$\signtutte(q,\gamma)$.}
{A graph $G=(V,E)$.}
{Determine whether the
sign of $Z(G;q,\gamma)$ is positive, negative, or $0$.}
\prob{$\signtutte(q;\gamma_1,\ldots,\gamma_k)$.}
{A graph $G=(V,E)$ and
a weight function $\bgamma:E\rightarrow \{\gamma_{1},\ldots,\gamma_{k}\}$.}
{Determine whether the sign of $Z(G;q,\bgamma)$ is positive, negative, or $0$.}

\subsection{Randomised algorithms and approximate counting}\label{sec:ranAlg}

A \emph{randomised algorithm} for a computational problem
takes an instance of the problem and returns a result.
We require that for each instance, and each run of the algorithm, the probability
that the result is equal to the correct output 
for the given instance
is at least~$\tfrac34$.

A \emph{randomised approximation scheme\/} is an algorithm for
approximately computing the value of a function~$f:\Sigma^*\rightarrow
\mathbb{R}$. The
approximation scheme has a parameter~$\varepsilon>0$ which specifies
the error tolerance.
A \emph{randomised approximation scheme\/} for~$f$ is a
randomised algorithm that takes as input an instance $ x\in
\Sigma^{* }$ (e.g., an encoding of a graph~$G$) and an error
tolerance $\varepsilon >0$, and outputs a number $z\in \mathbb{Q}$
(a random variable of the ``coin tosses'' made by the algorithm)
such that, for every instance~$x$,
\begin{equation*}
\Pr \big[e^{-\epsilon} \leq z/f(x) \leq e^\epsilon \big]\geq \frac{3}{4}\, ,
\end{equation*}
where, by convention, $0/0=1$.
(The slight modification of the more familiar definition is to 
ensure that functions~$f$ taking negative values are dealt with correctly.)

The randomised approximation scheme is said to be a
\emph{fully polynomial randomised approximation scheme},
or \emph{FPRAS},
if it runs in time bounded by a polynomial
in $ |x| $ and $ \epsilon^{-1} $.

Completeness of a problem in \#P is defined with respect to polynomial-time 
Turing reduction.   Suppose $\signtutte(q,\gamma)$ is \#P-hard for some 
setting of the parameters $q,\gamma$.  Then, clearly,
$\signtutte(q,\gamma)\in\FP$ would imply $\numP=\FP$.  In addition,
the existence of a polynomial-time randomised algorithm for $\signtutte(q,\gamma)$
would imply the existence of a polynomial-time randomised algorithm for every
problem in~$\numP$.  The reasoning is as follows.  Suppose the randomised
algorithm for $\signtutte(q,\gamma)$ has failure probability at most $\frac14$.
By a standard powering argument, the failure probability can be reduced 
so that it is exponentially small in the input size.  But a polynomial-time Turing  
reduction makes only polynomially many oracle calls, so the probability that even
a single one produces the wrong answer is exponentially small, and certainly less than~$\frac14$.
As an immediate consequence, an FPRAS for $\tutte(q,\gamma)$ would again 
imply the existence of a polynomial-time randomised (but exact in the event of success) 
algorithm for every problem in~$\numP$.

\section{\#P-hardness of computing the sign of the Tutte polynomial --- the multivariate case}
\label{sec:main}

We use the fact that the following problem is \#P-complete.
This was shown by
Provan and Ball~\cite{ProvanBall}. 
\prob{\MCCut.}
{A graph $G=(V,E)$ and distinguished vertices~$s,t\in V$.}
{$|\{S\subseteq E:\mbox{$S$ is a minimum cardinality $(s,t)$-cut in $G$}\}|$.}
 
 \setcounter{theorem}{1}
 \begin{lemma}
 \label{lem:signhardqbig} 
 Suppose $q>1$ 
 and  that 
 $\gamma_1\in(-2,-1)$
 and 
 $\gamma_2\notin[-2,0]$.
Then  
$\signtutte(q;\gamma_1,\gamma_2)$
is $\numP$-hard. 
\end{lemma}

\begin{proof} 
We will give a Turing reduction from \MCCut\ 
to $\signtutte(q;\gamma_1,\gamma_2)$.
 
Let $G,s,t$ be an instance of \MCCut.
Assume without loss of generality that $G$ has no edge from~$s$ to~$t$.
Let $n=|V(G)|$ and $m=|E(G)|$.
Assume without loss of generality that 
$G$ is connected and that 
$m\geq n$ is sufficiently large.
Let $k$ be the size of a minimum cardinality $(s,t)$-cut in~$G$ and
let $C$ be the number of size-$k$
$(s,t)$-cuts.

The following calculations are more general than necessary so that we can re-use them in the
proof of Lemma~\ref{lem:signhardqsmall} (where $q<1$ and $q$ may even be negative).
Let 
$$M^* = \max\left(   {\left(8 \max\left(|q|,\frac{1}{|q|}\right)\right)}^m, 
\frac{2}{|q-1|}\right).$$
Let $h$ be the smallest integer such that $(\gamma_2+1)^h-1>  M^*$
and let $M = (\gamma_2+1)^h-1$.
Note that we can implement~$M$ from~$\gamma_2$ via an $h$-thickening, and
$h$ is at most a polynomial 
in~$m$. 
 
Let $\delta =  
   {\left(2 \max(|q|,|q|^{-1})\right)}^m / M$. 
Let $\bM$ be the constant weight function which gives every
edge weight~$M$.
We will use the following facts:
\begin{equation} 
M^m q - \delta M^m |q| \leq Z_{st}(G;q,\bM) \leq M^m q + \delta M^m |q|
\label{thatitem}
\end{equation}
and
\begin{equation}
C M^{m-k}q^2(1-\delta) \leq Z_{s|t}(G;q,\bM) \leq C M^{m-k} q^2 (1+\delta).
\label{thisitem}
\end{equation}
Fact~(\ref{thatitem}) follows from
the fact that each of the (at most 
$2^m$) terms 
in $Z_{st}(G;q,\bM)$, 
other than the term with all edges in~$A$, has absolute value at most 
$M^{m-1} {\left(\max(|q|,1)\right)}^n$
and 
$ {2^m M^{m-1}{\left(\max(|q|,1)\right)}^n} \leq \delta
{M^m |q|}$.
Fact~(\ref{thisitem}) follows from
the fact that all terms in $Z_{s|t}(G;q,\bM)$ are complements of $(s,t)$-cuts.
Each term that is not a 
complement of a
size-$k$ $(s,t)$-cut has absolute value at most
$M^{m-k-1} q^2{\left(\max(|q|,1)\right)}^n$
and 
$$2^m
M^{m-k-1} 
q^2{\left(\max(|q|,1)\right)}^n \leq \delta C M^{m-k}q^2.$$
 
For a parameter~$\epsilon$ in the open interval $(0,1)$ which we will tune below, let
$\gamma' = -1-\epsilon \in (-2,-1)$. 
We will discuss the implementation of $\gamma'$ below.
Let $G'$ be the  graph formed from $G$ by adding an edge from $s$ to~$t$.
Let $\bgamma$ be the edge-weight function for~$G'$ that assigns weight~$M$ to every edge of~$G$ and
assigns weight $\gamma'$ to the new edge.
Then, using the definition of the Tutte polynomial,
\begin{align}
\nonumber
Z(G';q,\bgamma) &= Z_{st}(G;q,\bM)(1 + \gamma') + Z_{s|t}(G;q,\bM) \left(1 + \frac{\gamma'}{q}\right) \\
&= - \epsilon Z_{st}(G;q,\bM)  + Z_{s|t}(G;q,\bM) \left(1 - \frac{ 1+\epsilon}{q}\right).
\label{eq:Z}
\end{align}
 
Now suppose 
$\epsilon=M^{-2m}$. Then
$$
Z(G';q,\bgamma) = - M^{-2m} Z_{st}(G;q,\bM)  + Z_{s|t}(G;q,\bM) \left(1 - \frac{ 1+ M^{-2m}}{q}\right).
$$
Now since $M> 2/(q-1)$
and $M\geq1$, 
we have $1-(1+M^{-2m})/q \geq (1-1/q)/2$.  
(Note that $M$ is bounded away from~1, so $M^{-2m}$ can be made a small as 
we need by taking $m$ sufficiently large.)
So, using (\ref{thatitem}) and (\ref{thisitem}),
$$Z(G';q,\bgamma) \geq ((1-1/q)/2) C M^{m-k}q^2 (1-\delta) - M^{-2m} M^{m} q(1+\delta),$$
which is positive since $k \leq m$. 
On the other hand, 
using the definition of~$M$ and Facts (\ref{thatitem}) and (\ref{thisitem}) above,
we can confirm that   when $\epsilon=1$, $Z(G';q,\bgamma)$ is negative.
Also, when $\epsilon=q-1$ we have $Z(G';q,\bgamma) = - (q-1) Z_{st}(G;q,\bM)$,
which again is negative.

Thus we have a range from $\epsilon=M^{-2m}$ to $\epsilon=\min(1,q-1)$ 
of length at most~$1$ 
in which $Z(G';q,\bgamma)$ changes sign.
The idea is to perform binary search on this range
to find 
an 
$\epsilon$ where $Z(G';q,\bgamma)=0$.
For this value of~$\epsilon$, we have
$\epsilon Z_{st}(G;q,\bM)
= Z_{s|t}(G;q,\bM) \left(1 - \frac{ 1+\epsilon}{q}\right)$.
It turns out that, given this identity, estimates (\ref{thatitem}) and (\ref{thisitem}) 
above will give us enough information to calculate~$C$.

As one would expect, there are small technical complications. 
Since we are somewhat constrained in what values~$\epsilon$ we can implement, we won't be able to
discover the exact value of~$\epsilon$ that we need, but we will be able to approximate it sufficiently closely
to compute~$C$ exactly from (\ref{thatitem}) and~(\ref{thisitem}).
Suppose for a moment that we are able, for a given~$\epsilon\in(M^{-2m},\min(1,q-1))$,
to compute the sign of $Z(G';q,\bgamma)$.
Our basic strategy will be binary search, sub-dividing the
initial interval $\lceil m^2 \lg M \rceil$
times, so eventually we'll get an interval of width at most $M^{-m^2}$
which contains an $\epsilon$ where $Z(G';q,\bgamma)=0$.
 
To do this, we need to address the issue of computing 
the sign of 
$Z(G';q,\bgamma)$ using an oracle for
$\signtutte(q;\gamma_1,\gamma_2)$.
We have already seen above that it is easy to implement the weight~$M$ using~$\gamma_2$ 
(and that the implementation has polynomial size)  --- we now need to consider the implementation of
$\gamma'=-1-\epsilon$ (where $\epsilon\in(M^{-2m},\min(1,q-1))$ is the particular value that is being queried).

Let $y'=-\epsilon$ be the point that we desire to implement.
Let $y_1 = \gamma_1+1$. Note that $y_1 \in(-1,0)$.
Let $j$ be the smallest odd integer so that ${|y_1|}^j< \epsilon$.
Let $T^- = {|y_1|}^{-2}$ and $T^+ = {|y_1|^{-3}}$.
Let $T = -\epsilon/y_1^{j+2}$. Note that $1<T^- \leq T \leq T^+$.
   
Let $(x_2,y_2) = (q/\gamma_2+1,\gamma_2+1)$. Note that $y_2\notin[-1,1]$.
We will define a small quantity~$\pi$ below.
Looking ahead to  Lemma~\ref{lem:shiftbigqbigT}, 
we see that, from the point $(x_2,y_2)$ 
we can implement a point $(x'',y'')$ with
$T-\pi \leq y'' \leq T$.
The size of the graph used to implement $(x'',y'')$ is at most a polynomial in $\log(\pi^{-1})$.
It does not depend upon~$T$, though it does depend on the fixed bounds~$T^-$ and~$T^+$.
Now implement $y'$ by a parallel composition of $y''$ and $j+2$ copies of $y_1$.
(We can do this parallel composition because $j$ is only polynomially big in 
 $m$.)
Note that $-\epsilon  \leq y' \leq - \epsilon + \pi |y_1|^{j+2}  $, so of
course $- \epsilon \leq y' \leq - \epsilon + \pi$.

Thus, in the binary search, we may not be  able to
query the exact value of~$\epsilon$ that we want to, but we can query a value that 
is
between
$\epsilon-\pi$ and $\epsilon$.

Recall that our goal is to end up with a sub-interval
of the initial interval $(M^{-2m},\min(1,q-1))$  such that the subinterval has width at most $M^{-m^2}$
 and contains an $\epsilon$ where $Z(G';q,\bgamma)=0$.
We do this by setting $\pi = M^{-m^2}/3$ so that $\pi$ is only a third as large as the smallest
subinterval width (where we stop the binary search). We also adjust the binary search,
sub-dividing the original interval up to $\lceil m^2 \log_{3/2} M \rceil$ times rather than 
$\lceil m^2 \log_{2} M \rceil$ times, to make up for the fact that we might end up
with (crudely) at most two-thirds of the interval after one iteration, rather than half.
The result, then, is that 
we can find a subinterval of width at most $M^{-m^2}$
which contains an $\epsilon$ where $Z(G';q,\bgamma)=0$.

Now let $\epsilon$ be an endpoint of this subinterval.
Let 
$$\rho = 2^m {\max(|q|,1)}^m M^m M^{-m^2}.$$ 
Since $\epsilon \geq M^{-2m}$
and $m$ is sufficiently large, 
we have
$\rho \leq \epsilon M^m |q| 
4^{-m}$.
From the definition of Tutte polynomial, $Z(G';q,\bgamma)$ is linear as a function of $\gamma'$ (and hence of $\epsilon$), 
and the coefficient of $\gamma'$ is a sum of $2^m$ terms, each bounded in absolute value 
by $\max(|q|,1)^nM^m\leq\max(|q|,1)^mM^m$.   Since $\gamma'$ is within distance $M^{-m^2}$ of the zero of $Z(G';q,\bgamma)$,
we see that $|Z(G';q,\bgamma)| \leq  \rho$. 

Now using (\ref{eq:Z}), (\ref{thatitem}) and (\ref{thisitem}), 
we have
$$
\frac{-\rho + \epsilon M^m q(1-\delta)}{\left(1-\frac{1+\epsilon}{q}\right) M^{m-k} q^2 (1+\delta)}
\leq C \leq 
\frac{\rho + \epsilon M^m q(1+\delta)}{\left(1-\frac{1+\epsilon}{q}\right) M^{m-k} q^2(1-\delta)}.  
$$ 
so, since $\delta \leq 
4^{-m}$,
\begin{equation}
\label{lastone}
\frac{(1-2\cdot4^{-m})\epsilon M^m q}{\left(1-\frac{1+\epsilon}{q}\right) M^{m-k} q^2 (1+4^{-m})}
\leq C \leq 
\frac{  \epsilon M^m q(1+ 2 \cdot 4^{-m})}{\left(1-\frac{1+\epsilon}{q}\right) M^{m-k} q^2(1-4^{-m})}.  
\end{equation}

Now the point is that $C$ is an integer between~$1$ 
and~$2^m$.
Even though the value of~$k$ is not known, 
the fact that $M>4^m$ means that there can only be one integer~$k$
such that the above interval contains an integer between~$1$ and~$2^m$
(so $k$ can easily be deduced). All of the other quantities in the lower and upper bounds in 
(\ref{lastone}) are known.
Now let 
$R = \frac{\epsilon M^k }{q-(1+\epsilon)}$, so
(\ref{lastone}) becomes
\begin{equation}
\label{lastlast}
\left(\frac{1-2\cdot 4^{-m}}{  1+4^{-m}} \right)R
\leq C \leq 
   R 
\left( \frac{1+ 2 \cdot 4^{-m}}{1-4^{-m}}\right) . 
\end{equation}

Now, 
$R<2^{m+1}$, 
since otherwise the left-hand-side of (\ref{lastlast}) is greater than~$2^m$.
Also, 
multiplying through by $(1+4^{-m})(1-4^{-m})$,
the width of the interval  is at most 
$6 \cdot 4^{-m} R < 1$ 
so the width of the interval in (\ref{lastlast}) is less than~$1$, so 
the (integral) value of~$C$ can be calculated exactly.
\end{proof}

 We have a similar lemma for $q<1$.
 
 \begin{lemma}
 \label{lem:signhardqsmall}  
 Suppose $q<1$ and $q\neq 0$ and  that 
 $\gamma_1\in(-1,0)$
 and 
 $\gamma_2\notin[-2,0]$.
Then  
$\signtutte(q;\gamma_1,\gamma_2)$
is $\numP$-hard. 
\end{lemma}
\begin{proof}

The situation is very similar to that of Lemma~\ref{lem:signhardqbig}.

We start with the situation $0<q<1$.
In this case, we follow the proof of Lemma~\ref{lem:signhardqbig}.
Then Facts~(\ref{thatitem}) and~(\ref{thisitem}) hold, as before.
For the tuneable parameter $\epsilon\in(0,1)$,  
we let $\gamma' = -1+\epsilon\in (-1,0)$.
Implementing $G'$ as  in the proof of Lemma~\ref{lem:signhardqbig}, we have
\begin{equation}
\label{newqeq:Z}
Z(G';q,\bgamma) =  \epsilon Z_{st}(G;q,\bM)  + Z_{s|t}(G;q,\bM) \left(1 - \frac{ 1-\epsilon}{q}\right) .
\end{equation}
Now,  suppose $\epsilon=M^{-2m}$.
Then since $M>2/(1-q)$ and $M\geq1$, we have
$$1-(1-M^{-2m})/q \leq \tfrac12(1-1/q)<0.$$
Using Facts~(\ref{thatitem}) and (\ref{thisitem}), we find that 
$Z(G';q,\bgamma)$ is negative.
On the other hand, at $\epsilon=1-q$, $Z(G';q,\bgamma)$ is positive.

To implement $\gamma'$, let $y'=\epsilon$ be the point that we desire to implement.
Let $y_1=\gamma_1+1$. Note that $y_1 \in (0,1)$.
Now proceed as in the proof of Lemma~\ref{lem:signhardqbig},
with $T=\epsilon/y_1^{j+2}$, and $T^-$ and $T^+$ as before.
Once again
we find a subinterval of $( M^{-2m},1-q)$
of width at most $M^{-m^2}$ which contains an $\epsilon$ where $Z(G';q,\gamma)=0$, so
we let $\epsilon$ be an endpoint of this subinterval and we conclude
that $|Z(G';q,\gamma)| \leq \rho$.
Now we finish as in the proof of Lemma~\ref{lem:signhardqbig}.
 
The argument for $q<0$ also follows the proof of Lemma~\ref{lem:signhardqbig}.
Here,  $Z_{st}(G;q,\bM)$ is negative
and  
$Z_{s|t}(G;q,\bM)$ is positive.
Taking~$\gamma'=-1+\epsilon$, as above,
we still have (\ref{newqeq:Z}).
Now, suppose $\epsilon=M^{-2m}$.
Then by (\ref{newqeq:Z}), 
$Z(G';q,\bgamma) \geq M^{-2m} Z_{st}(G;q,\bM)+ Z_{s|t}(G'q,\bM)$, which is positive.
On the other hand, at $\epsilon=1$,
$Z(G';q,\bgamma)$ is negative.
Now the implementation of~$\gamma'$ proceeds as above, except that we use 
 Lemma~\ref{lem:shiftsmallqbigT} (working from points $(x_1,y_1)$ and $(x_2,y_2)$)
 instead of Lemma~\ref{lem:shiftbigqbigT}. 
 
So
we find a subinterval of $( M^{-2m},1)$
of width at most $M^{-m^2}$ which contains an $\epsilon$ where $Z(G';q,\gamma)=0$. 
Letting 
$\epsilon$ be an endpoint of this subinterval, we conclude
that $|Z(G';q,\gamma)| \leq \rho$. 
Now we finish as in the proof of Lemma~\ref{lem:signhardqbig}.
\end{proof}

\section{Implementing new edge weights}
\label{sec:shift}

In this section, we collect the information that we need about
implementing edge weights within various regions of the Tutte Plane.  
The following straightforward lemmas are useful.

 \begin{lemma}
 \label{lem:xlefttoyup}
 Suppose $q>0$ and that $(x,y)$ is a point with $x<-1$. 
 Then $(x,y)$ can be used to implement a point $(x',y')$ with $y'>1$.
 \end{lemma}
 \begin{proof}
 A $2$-stretch from $(x,y)$ suffices
 since it implements the point $(x',y')=(x^2,(x+y)/(1+x))$.
 If $x<-1$ and $q=(x-1)(y-1)$ is positive then $y<1$ so $x+y$ and $1+x$ are both negative.
 Since $y<1$ we conclude that $-y>-1$ so $-x-y>-1-x$ and $y'>1$.
 \end{proof}

We will use the following Lemma, which is \cite[Lemma 3.26]{planartutte}.
The lemma in~\cite{planartutte} was stated for $q>5$ (which was all that was needed in that paper) but the proof only uses $q>0$.
The statement in~\cite{planartutte} was in terms of the coordinates $q$ and $\gamma$ but
we have translated it to $(x,y)$ coordinates here, since that is how it will be used here.
Finally, the statement of the lemma in~\cite{planartutte} allowed the implementation to use two additional
points~$(x'_2,y'_2)$ and~$(x'_3,y'_3)$ (this was to make the statement of the lemma match other
lemmas in that paper). However, these additional points were not used in the proof, so we don't include them here.
 
 \begin{lemma}  (\cite[Lemma 3.26]{planartutte}) \label{lem:shiftbigqbigT}
Suppose that $(x_1,y_1)$ is a point with
$y_1\notin[-1,1]$ and that $q=(x_1-1)(y_1-1)>0$.
Suppose that $T^-$ and $T^+$ satisfy 
$1<T^-\leq T^+$.
Given a target edge-weight $T\in[T^-,T^+]$ and a
positive value~$\pi$ which is sufficiently small with respect to~$x_1$, $y_1$, $T^-$ and $T^+$,
a point $(x,y)$ with $T- \pi \leq y \leq T$ can be
implemented using the point
$(x_1,y_1)$. 
The size of the graph~$\Upsilon$ used to implement $(x,y)$  is at most a polynomial in $\log(\pi^{-1})$.
(This upper bound on the size of $\Upsilon$ does not depend on~$T$, though it
does depend on the fixed bounds~$T^-$ and~$T^+$.)
\end{lemma}

By duality of~$x$ and~$y$, we have the following corollary.
 \begin{corollary}  \label{cor:shiftbigqbigT}
Suppose that $(x_1,y_1)$ is a point with
$x_1\notin[-1,1]$ and that $q=(x_1-1)(y_1-1)>0$.
Suppose that $T^-$ and $T^+$ satisfy 
$1<T^-\leq T^+$.
Given a target edge-weight $T\in[T^-,T^+]$ and a
positive value~$\pi$ which is sufficiently small with respect to~$x'_1$, $y'_1$, $T^-$ and $T^+$,
a point $(x,y)$ with $T- \pi \leq x \leq T$ can be
implemented using the point
$(x_1,y_1)$. 
The size of the graph~$\Upsilon$  used to implement $(x,y)$  is at most a polynomial in $\log(\pi^{-1})$.
(This upper bound on the size of $\Upsilon$ does not depend on~$T$, though it
does depend on the fixed bounds~$T^-$ and~$T^+$.)
\end{corollary}

We will also use the following related lemma, which is \cite[Lemma 3.27]{planartutte}.
Once again, we translated to $(x,y)$ coordinates
and eliminated unused points.

 \begin{lemma}  (\cite [Lemma 3.27]{planartutte}) \label{lem:shiftsmallqbigT}
Suppose that $(x_1,y_1)$ is a point with $y_1 \notin[-1,1]$
and $(x_2,y_2)$ is a point with $y_2 \in (-1,1)$.
Suppose that $q=(x_1-1)(y_1-1)=(x_2-1)(y_2-1)<0$.
Suppose that $T^-$ and $T^+$ satisfy 
$1<T^-\leq T^+$.
Given a target edge-weight $T\in[T^-,T^+]$ and a
positive value~$\pi$ which is sufficiently small with respect to~$x_1$, $y_1$, $x_2$, $y_2$, 
$T^-$ and $T^+$,
a point $(x,y)$ with $T-\pi \leq y \leq T$ can be implemented using the points
$(x_1,y_1)$ and $(x_2,y_2)$.
The size of the graph~$\Upsilon$ used to implement $(x,y)$  is at most a polynomial in $\log(\pi^{-1})$.
(This upper bound on the size of $\Upsilon$ does not depend on~$T$, though it
does depend on the fixed bounds~$T^-$ and~$T^+$.)
\end{lemma}

The reader may find it useful to
consult Figure~\ref{fig:one}, 
and the formal definitions listed early in Section~\ref{sec:red},
to see the relevant regions of the~$(x,y)$ plane that we consider.
 
\subsection{Region B}

 The following four lemmas prepare the conditions for applying Lemma~1 to points in Region~B.
 Note that  the value $q=(x-1)(y-1)$ exceeds~$1$ in this region.
 
\begin{lemma} 
\label{lem:shiftB1}
Suppose $(x,y)$ is a point with $x<-1$ and $y<-1$.
Then we can use $(x,y)$ to implement a point 
$(x_1,y_1)$ with $y_1 \in (-1,0)$
and a point $(x_2,y_2)$ with $y_2 \notin [-1,1]$.
\end{lemma}
\begin{proof}
Let $q=(x-1)(y-1)$.
Let $j$ be an odd positive integer which is sufficiently large that $|x|^j+1>q$.
Implement $(x',y') = (x^j, q/(x^j-1)+1)$ from $(x,y)$ with a $j$-stretch.
Note that $y'\in(0,1)$.
Now, for a sufficiently large positive integer~$k$, 
implement $(x_1,y_1)$ using the parallel composition of $(x,y)$ with $k$ copies of $(x',y')$
so $y_1 = {y'}^k y \in (-1,0)$.
Finally, let $(x_2,y_2) = (x,y)$.
\end{proof}

\begin{lemma} 
\label{lem:shiftB2}
Suppose $(x,y)$ is a point with $x<-1$ and $y=-1$.
Then we can use $(x,y)$ to implement a point 
$(x_1,y_1)$ with $y_1 \in (-1,0)$
and a point $(x_2,y_2)$ with $y_2 \notin [-1,1]$.
\end{lemma}
\begin{proof}  
Let $j$ be a sufficiently large odd integer such that
$q/(|x|^j+1)<1$.
Implement   $(x',y')$ using a $j$-stretch from $(x,y)$
so that $y' = q/(x^j-1)+1 \in (0,1)$. Implement $(x_1,y_1)$ by taking a parallel composition of $(x',y')$ and $(x,y)$
so $y_1 = -y'$. Finally, implement $(x_2,y_2)$ from $(x,y)$ using Lemma~\ref{lem:xlefttoyup}.
\end{proof}

\begin{lemma} 
\label{lem:shiftB3}
Suppose $(x,y)$ is a point with $x<-1$ and $-1<y<0$.
Then we can use $(x,y)$ to implement a point 
$(x_1,y_1)$ with $y_1 \in (-1,0)$
and a point $(x_2,y_2)$ with $y_2 \notin [-1,1]$.
\end{lemma}
\begin{proof}
We let $(x_1,y_1)=(x,y)$.
We implement $(x_2,y_2)$ from $(x,y)$ using Lemma~\ref{lem:xlefttoyup}.
\end{proof}

\begin{lemma} 
\label{lem:shiftB4}
Suppose $(x,y)$ is a point with $-1\leq x <0$ and $y<-1$.
Then we can use $(x,y)$ to implement a point 
$(x_1,y_1)$ with $y_1 \in (-1,0)$
and a point $(x_2,y_2)$ with $y_2 \notin [-1,1]$.
\end{lemma}
\begin{proof}
Implement
$(x_a,y_a)$ by a $2$-thickening of $(x,y)$.
Note that $y_a=y^2>1$, and therefore, since $q>0$, $x_a>1$ as well.
Let $j$ be an integer that is sufficiently large that $|x| \cdot x_a^j+1 > q$.
Implement $(x_b,y_b)$ by a series composition of $(x,y)$ with $j$
copies of $(x_a,y_a)$ so that
$$y_b = q/(x x_a^j-1)+1 \in (0,1).$$
Let $k$ be a sufficiently large integer that $|y| y_b^k \in (0,1)$.
Implement $(x_1,y_1)$ by a parallel composition of $(x,y)$ and
$k$ copies of $(x_b,y_b)$ 
so $y_1 = y y_b^k$.
Finally, let $(x_2,y_2)=(x,y)$.
\end{proof}  
   
\subsection{Regions G, H and I}

We next consider the problem of implementing edge weights
starting from a point in the ``vicinity of the origin'',
which corresponds to points with $|x|<1$ and $|y|<1$.
In the vicinity of the origin, we have $0<q<4$.
As noted in the introduction, there is a ``phase transition'' at $q=32/27$, so
we start by considering $q>32/27$.

\begin{lemma}
\label{lem:3227g2} 
Suppose 
$(x,y)$ is a point with 
$|x| < 1$ and $|y| < 1$ and
$q=(x-1)(y-1)>32/27$.
Then $(x,y)$ can be used to implement 
a point $(x',y')$ with $y'>1$.
\end{lemma}

\begin{proof}
 
We will use the 
``diamond operation'' of Jackson and Sokal~\cite[Section 8]{JacksonSokal}.
This   corresponds
to choosing the graph 
$\Upsilon$ with vertex set $\{s,t,u,v\}$
and edge set $\{(s,u),(u,t),(s,v),(v,t)\}$.
($\Upsilon$ is a 
parallel composition of two paths from~$s$ to~$t$, 
each of which is formed from the series composition of two edges.
If we start with the weight function $\hatbgamma$
that assigns weight~$\gamma$ to every edge of~$\Upsilon$, then
it is easy to check (see \cite[(8.1)]{JacksonSokal}) that the
implemented weight~$w^*$ from Equation~(\ref{eq:implement}) is
$\frac{\gamma^2(\gamma^2 + 4 \gamma + 2q)}{{(q+2\gamma)}^2}$. 
Equivalently, 
the point 
implemented from $(x,y)$ 
(which we denote as $(\diamond_{q,1}(x,y),\diamond_{q,2}(x,y))$) 
is
given by
$$\left(\diamond_{q,1}(x,y),\diamond_{q,2}(x,y)\right) =
\left(\frac{x+x^2+x^3+y}{1+2x+y},\frac{(x+y)^2}{(1+x)^2}\right).$$

The diamond operation is well-defined as
long as $x\neq -1$ and $y\neq -1-2x$.
Jackson and Sokal~\cite[Lemma 8.5(c)]{JacksonSokal} prove
that if you start from a point $(x_1,y_1)$ with $y_1<1$ and $q>32/27$
and apply a sequence of diamond operations for $j=1,2,\ldots$
with $(x_{j+1},y_{j+1}) =  (\diamond_{q,1}(x_{j},y_{j}),\diamond_{q,2}(x_{j},y_{j}))$
then for each $j$, we have $y_{j+1}>y_{j}$ and  there is a $k$ such that $y_k\geq 1$. 
Their analysis allows the 
situation $x_{k-1}=-1$, so the terminating point has $y_k=\infty$
(which would not give an implementation of a \emph{finite} $y_k>1$, which we require)
and it also allows $y_{k-1}=-1-2 x_{k-1}$
which gives $y_k=1$ (whereas we require $y_k>1$).

We start with $(x_1,y_1)=(x,y)$ and apply the sequence of diamond operations
until we reach a point $(x_j,y_j)$ 
with $y_j>1$.
However, there are two exceptions.

First, suppose, for some~$j$,  that $y_j=-1-2x_j$.
Then instead of taking
$(x_{j+1},y_{j+1}) =  (\diamond_{q,1}(x_{j},y_{j}),\diamond_{q,2}(x_{j},y_{j}))$
we define $(x_{j+1},y_{j+1})$ as follows:
We let $(x'_1,y'_1)=(x_j^2,-1)$ be the point implemented by a series composition of two copies of $(x_j,y_j)$.
We then let $(x'_2,y'_2)=(x_j^4,(x_j^2-1)/(x_j^2+1))$ be the point implemented by a series composition of two copies of $(x'_1,y'_1)$.
Finally, we let $(x_{j+1},y_{j+1}) = (1-x_j^{-2}+x_j^2,(1-x_j^2)/(1+x_j^2))$
be the point implemented by a parallel composition  of $(x'_1,y'_1)$ and $(x'_2,y'_2)$.
Note that $y_{j+1}-y_j = 2(x_j^3 + x_j + 1)/(x_j^2 + 1)$.
Now note that $q=2-2x_j^2$ so, since $q\geq 32/27$, we have  $x_j>-0.64$.
Thus, $y_{j+1}-y_j$ is positive, as required (the denominator is always positive,
and the numerator is positive for $x_j\geq -0.68$).
Note that exceptional points $(x_j,y_j)$ where $y_j=-1-2x_j$ arise at most twice during
the  sequence of points $(x_1,y_1)$, $(x_2,y_2)$, $\ldots$
since the hyperbola $(x-1)(y-1)=q$ only intersects the line $y=-1-2x$ in at most two places.
Also, $y_{j+1}\neq 1$, so the sequence does not terminate  incorrectly at $(x_{j+1},y_{j+1})$.

For the second exception, suppose that we get to a point $(x_j,y_j)$ with $x_j=-1$.
Then
$(x_j,y_j) = (-1,-q/2+1)$.  
Now, $j\neq 1$ since we start in the vicinity of the origin (so we don't have $x_1=-1$).
If $(x_j,y_j)$ was obtained as a result of the exceptional case above, then $q<2$
(since then $q=2-2 x_{j-1}^2$ and $x_{j-1}\neq 0$ since that would imply $y_{j-1}=-1$, contrary
to the fact that the $y$'s are all strictly above $-1$).
Otherwise,
$(x_j,y_j)$ was
obtained as the result of a diamond operation.
It is not possible that $x_{j-1}=-y_{j-1}$ since 
then $q=(x_{j-1}-1)(y_{j-1}-1)=-x_{j-1}^2+1\leq 1$.
Thus, from the definition of the diamond operation, $y_j>0$.
Thus, since $y_j=-q/2+1$, we also have
$q<2$.
Let $(x^*,y^*)$ be obtained as a parallel composition of two copies of $(x_j,y_j)$;
then $(x'',y'')$ as a series composition of $(x_j,y_j)$ and $(x^*,y^*)$.
By direct calculation from the series/parallel formulas,
$$
(x^*,y^*)=\left(\frac{-q}{4-q},\frac{(q-2)^2}4\right)\quad\text{and}\quad
(x'',y'')=\left(\frac{q}{4-q},\frac{q^2-6q+4}{2(2-q)}\right).
$$
It can be verified that $y''<-1$ in the range $32/27\leq q<2$.  ($y''$
is monotonically decreasing in~$q$, and less than $-1$ at $q=32/27$.)
So letting $(x',y')$ be a parallel composition of two copies of $(x'',y'')$
we are done, since $y'>1$.
\end{proof}

 \begin{lemma} 
\label{lem:triangle1}
Consider a point $(x,y)$ such that $y<-1-2x$ and $x> -1$.
 Then $(x,y)$ can be used to implement 
 a point $(x',y')$ with $y'>1$. 
 \end{lemma} 
\begin{proof}
Let $(x'',y'')= (x^2,\frac{x+y}{1+x}$) be the point implemented by a $2$-stretch from $(x,y)$.
Note that $y''<-1$. Now implement $(x',y')$ by a $2$-thickening of $(x'',y'')$.
\end{proof}

 \begin{lemma} 
\label{lem:triangle2}
Consider a point $(x,y)$ such that
 $x<-1-2y$ and $y> -1$ and $q=(x-1)(y-1)>0$. 
 Then $(x,y)$ can be used
to implement 
a point $(x',y')$ with $y'>1$. 
\end{lemma} 
\begin{proof}
Let $(x',y') = (\frac{x+y}{1+y},y^2)$ be the point implemented by a $2$-thickening.
Note that $x'<-1$. Then use Lemma~\ref{lem:xlefttoyup}.
\end{proof}

\begin{lemma}
\label{lem:3227g1} 
Suppose  that $(x,y)$ is a point satisfying
$\max(|x|,|y|)<1$
  and $q=(x-1)(y-1)>1$.
Suppose that $(x,y)$ also satisfies at least one of the following conditions.
\begin{itemize}
\item $q>32/27$, or
\item $y<-1-2x$, or
\item $x<-1-2y$. 
\end{itemize}
Then $(x,y)$ can be used to implement 
a point $(x_1,y_1)$ with $-1<y_1<0$.
\end{lemma}

\begin{proof} 
If $-1<y<0$ then we simply take $(x_1,y_1)=(x,y)$.
Thus, we can assume $0\leq y < 1$.
This implies $-1<x<0$, 
and $q>32/27$ or $y<-1-2x$.

By Lemmas~\ref{lem:3227g2} and~\ref{lem:triangle1},
we can implement a point $(x_1',y_1')$ with $y_1'>1$. Since $(x_1'-1)(y_1'-1)=q$, we
also have $x_1'>1$.

Note that the restrictions on~$x$ and~$y$ imply $1<q<2$.
Choose an even integer~$j$ so that $x^j < 1-q/4$.
By Corollary~\ref{cor:shiftbigqbigT}
(taking $T = (1-q/4)/x^j$ and $\pi=q/(8x^j)$, say) 
the point $(x_1',y_1')$
can be used to implement a point $(x'',y'')$ with
$$\frac{1-q/2}{x^j}< x'' < \frac{1}{x^j}.$$
Implement $(x^*,y^*)$ by taking the series composition of $(x'',y'')$ with $j$~copies of $(x,y)$.
Note that 
$y^* = \frac{q}{ x'' x^j -1}+1 < -1$.
  
Now implement $(x_1,y_1)$ by 
choosing a sufficiently large integer~$\ell$ and
taking the parallel composition of $(x^*,y^*)$ with
$\ell$ copies of $(x,y)$ so
that 
$y_1 = y^* y^\ell$.
\end{proof}

\subsection{Regions~C and~D}

\begin{lemma}
\label{lem:stretchCD}
Suppose $(x,y)$ is a point satisfying  
one of the following.
\begin{itemize}
\item $y>1$ and $x<-1$, or
\item $x>1$ and $y<-1$.
\end{itemize}
Then $(x,y)$ can be used to implement a point $(x_1,y_1)$ with 
$y_1\in (0,1)$.
\end{lemma}
\begin{proof}
Note that $q<0$. Choose an even number $j$ such that
$x^j-1>|q|$. Implement $(x_1,y_1)$ 
by taking a $j$-stretch of $(x,y)$ so
$y_1 =  q/(x^j-1)+1$.
\end{proof}

\subsection{Region E}

\begin{lemma}
\label{lem:Fq12}
Suppose $(x,y)$ is a point satisfying $x<-1$ and $0<y<1$ and
$1<(x-1)(y-1)<2$.
Then $(x,y)$ can be used to implement a point $(x_1,y_1)$ with $-1<y_1<0$.
\end{lemma}
\begin{proof}
Let $q=(x-1)(y-1)$.
Note that $1-q/2>0$ since $q<2$.
Let 
$j$ be a sufficiently large integer  that 
$0<y^j < 1-q/2$.
Note that $1-q<0$ so
$1-q < y^j < 1-q/2$.
Implement $(x',y')$ by $j$-thickening from the point~$(x,y)$ so
that $x'=q/(y^j-1)+1$. Note that  $-1<x'<0$. 
Now let $k$ be an odd integer which is sufficiently large that
$0 < x (x')^k < 1-q/2$ so $1-q < x (x')^k< 1-q/2$.
Implement $(x_1,y_1)$ by 
taking a series composition of $(x,y)$ with $k$ copies of $(x',y')$
so $x_1 = x (x')^k$.  Then $y_1 = q/(x (x')^k-1)+1$ so $-1<y_1<0$, as required.
\end{proof}

\begin{lemma}
\label{lem:F}
Suppose $(x,y)$ is a point satisfying $x<-1$ and $0<y<1$.
Suppose that $q=(x-1)(y-1)>2$ is not an integer.
Then $(x,y)$ can be used to implement a point $(x',y')$ with $y'<0$.
\end{lemma}

\begin{proof}
Let $q=(x-1)(y-1)$.
Let us first examine what points we can implement from the point 
$(x_1,y_1) = (1-q,0)$ 
and from points nearby.
We will later show how to implement points near $(x_1,y_1)$ from the given point~$(x,y)$.
Let $n=\lfloor q \rfloor + 2$. Note that $n\geq 4$ and that  $n-2<q<n-1$. 
Let $\Gamma_n$ be the graph obtained from the complete graph~$K_n$ 
on $n$ vertices by deleting some edge $(s,t)$.  
Let $\bgamma$ be the weight function that
gives every edge of $\Gamma_n$ 
weight~$ y_1-1=-1$.
From Section~\ref{sec:shiftdef}, the graph $\Gamma_n$ and the
weight function $\bgamma$ implement the weight
\begin{equation}
\label{defw} 
w(q,n) = 
\frac{q Z_{st}(\Gamma_n;q, -1)}{Z_{s|t}(\Gamma_n;q, -1)}.
\end{equation}

We wish to calculate some properties of $w(q,n)$.
Recall from the introduction that $Z(G;q,-1)$ is equal to the chromatic polynomial $P(G;q)$.
We will next calculate
$Z_{st}(\Gamma_n;q,-1)$
and $Z_{s|t}(\Gamma_n;q,-1)$ as polynomials in~$q$
using known facts about the chromatic polynomial.
In particular, when $q$ is a positive integer, $Z(G;q,-1)$ gives the number of proper $q$-colourings
of~$G$.

Now, let $V$ denote the vertex set of~$K_n$.
We can  expand the definition of $Z(K_n;q,-1)$
as 
$$Z(K_n;q,-1) = \sum_{A \subseteq E-(s,t)} \left(
q^{\kappa(V,A\cup \{(s,t)\})} {(-1)}^{|A|+1} + q^{\kappa(V,A)} {(-1)}^{|A|}
\right),$$
If a subset $A$ connects $s$ and $t$ then $\kappa(V,A\cup\{(s,t)\}) = \kappa(V,A)$
so the contribution from this~$A$ is zero. On the other hand, if
a subset $A$ does not connect $s$ and $t$ then
$\kappa(V,A\cup\{s,t\}) =  \kappa(V,A)-1$. Thus, 
\begin{align}
Z(K_n;q,-1) &= Z_{st}(\Gamma_n;q,-1)(1-1) + Z_{s|t}(\Gamma_n;q,-1)(1-\tfrac1q) \nonumber\\
                   &= Z_{s|t}(\Gamma_n;q,-1)(1-\tfrac1q). \label{eq:first}
\end{align}
Note that the factor $(1-\tfrac1q)$ is positive. 

Similarly,
$$Z_{st}(\Gamma_n;q,-1) = Z(\Gamma_n;q,-1) - Z_{s|t}(\Gamma_n;q,-1),$$
so we have
\begin{equation}
\label{eq:st}
Z_{st}(\Gamma_n;q,-1) = Z(\Gamma_n;q,-1) - \frac{Z(K_n;q,-1)}{1-\tfrac1q},
\end{equation}
and
\begin{equation}
\label{eq:sbart}
Z_{s|t}(\Gamma_n;q,-1) = \frac{Z(K_n;q,-1)}{1-\tfrac1q}.
\end{equation}
 
The properties of $w(q,n)$ that we require will follow from 
(\ref{defw}) and Equations~(\ref{eq:st}) and~(\ref{eq:sbart}).
First note that $Z(K_n;q,-1)=\prod_{i=0}^{n-1}(q-i)$.
This is clear at positive integer~$q$, since both sides can be in
interpreted as the number of $q$-colourings of an $n$-clique.
But we know that $Z(K_n;q,-1)$ is a polynomial in~$q$, so the two sides 
must be equal for all~$q$.
Let $N_{q,n} = \prod_{i=0}^{n-2}(q-i)$, so $Z(K_n;q,-1) = N_{q,n}(q-n+1)$.
Then $Z(\Gamma_n;q,-1) = N_{q,n}(q-n+2)$ since, again, both sides may be 
interpreted as the number of $q$ colourings of a certain graph, in this case~$\Gamma_n$. 
(If you colour the vertices of $\Gamma_n$ in order, colouring~$s$
last, there are $q-(n-2)$ choices for~$s$, rather than $q-(n-1)$ in $K_n$.)
Then, from~(\ref{defw}), (\ref{eq:st}) and~(\ref{eq:sbart}),
\begin{align*} 
w(q,n) =
\frac{q Z_{st}(\Gamma_n;q,-1) }{Z_{s|t}(\Gamma_n;q,-1)}
&= 
\frac{q Z(\Gamma_n;q,-1) - q\frac{Z(K_n;q,-1)}{1-1/q}}
{\frac{Z(K_n;q,-1)}{1-1/q}}  \\
&= 
\frac{(q-1) Z(\Gamma_n;q,-1) - q {Z(K_n;q,-1)}}
{{Z(K_n;q,-1)} } \\
&= 
\frac{(q-1) (q-n+2) - q(q-n+1)} 
{q-n+1} \\
&= 
\frac{   n-2}
{ q-n+1}, 
 \end{align*}
where we use the fact that $q$ is not integral, so $Z(K_n;q,-1)\not=0$.
 
Now since $n>2$ and $1<q<n-1$ we can see that the numerator
$n-2$ is positive, the denominator $q-n+1$ is negative
and $n-2> n-q-1$, and hence $w(q,n)<-1$.

 We now have 
 \begin{equation}
 \label{endlessedits}
 \frac{q Z_{st}(\Gamma_n;q,-1) }{Z_{s|t}(\Gamma_n;q,-1)}
 <-1.\end{equation}
 Unfortunately, we are not finished, because we cannot necessarily implement the
 weight $-1$ exactly from the given point~$(x,y)$.
 However, by continuity, Equation~(\ref{endlessedits}) implies 
 that there is a small positive $\epsilon$ (depending on $q$ and $n$)
such that, if
$|z- Z_{st}(\Gamma_n;q,-1)|\leq\epsilon$ 
and $|z'- Z_{s|t}(\Gamma_n;q,-1)|\leq\epsilon$, 
then we have
  $\frac{q  z}{ z'}
 <-1.$
 
To finish, we will show that we can implement an edge weight $-1+\delta$
from~$(x,y)$ so that 
 $|Z_{st}(\Gamma_n;q,-1+\delta)- Z_{st}(\Gamma_n;q,-1)|\leq\epsilon$ 
and $|Z_{s|t}(\Gamma_n;q,-1+\delta)- Z_{s|t}(\Gamma_n;q,-1)|\leq\epsilon$.
Thus, we can implement an edge-weight less than~$-1$ by
using~$\Gamma_n$ with all edge weights equal to $-1+\delta$.

We finish with the relevant technical details.
First, let $V$ be the vertex set of~$\Gamma_n$.
For any $\delta \in (0,\epsilon/(2^m q^n m))$,
note that
\begin{align*} 
Z_{st}(\Gamma_n;q,-1+\delta)- Z_{st}(\Gamma_n;q,-1)
&= \sum_{A } q^{\kappa(V,A)} {(-1)}^{|A|+1}
\left({1-(1-\delta)}^{|A|}  \right) \\
&\leq \sum_{A} q^{\kappa(V,A)}  
\left({1-(1-\delta)}^{|A|}  \right) \\
&\leq 2^m q^n  m \delta \\
&< \epsilon,
\end{align*}
where the sum is over edge subsets~$A$ with $s$ and $t$ in the same component.
Similarly, $Z_{st}(\Gamma_n;q,-1) -Z_{st}(\Gamma_n;q,-1+\delta) < \epsilon$
and 
 $|Z_{s|t}(\Gamma_n;q,-1+\delta)- Z_{s|t}(\Gamma_n;q,-1)|\leq\epsilon$.

It remains to show that we can implement weight $-1+\delta$ 
from the given $(x,y)$. 
Using $(x,y)$ coordinates, the point that we wish to implement is $(x'',y'')= (1+q/(\delta-1),\delta)$.
This can be done using a $k$-thickening from $(x,y)$, choosing $k$
to be sufficiently large that 
$y^k \leq \epsilon/(2^m q^n m)$.
\end{proof}

As we shall see shortly, Region~B consists of those points $(x,y)$ 
for which $\min(x,y) \leq -1$ and $\max(x,y)<0$. Also,
Region~G consists of points $(x,y)$ with $\max(|x|,|y|)<1$ and
$q=(x-1)(y-1)>32/27$.
We use these definitions in the following lemma.

\begin{lemma}
\label{lem2:F}
Suppose $(x,y)$ is a point satisfying $x<-1$ and $0<y<1$.
Suppose that $q=(x-1)(y-1)>2$ is not an integer.
Then $(x,y)$ can be used to implement a point $(x',y')$ 
apart from the special  point $(-1,-1)$ which is either in Region~B or in Region~G.
\end{lemma}
\begin{proof}
By Lemma~\ref{lem:F}, the point $(x,y)$ can be used to implement 
a point $(x',y')$ with $y'<0$.
We know that $(x',y')$ is not   the special point $(-1,-1)$ 
since $q$ is not an integer.
If $(x',y')$ is in Region~B or Region~G,
then we are finished.
Otherwise, the point $(x',y')$ satisfies  
$0\leq x'<1$ and
$y'\leq -1$.
Let $j$ be a sufficiently large integer so that
$|y'| y^j <1$.
Then implement the point $(x'',y'')$ by taking the parallel composition of $(x',y')$ with
$j$ copies of $(x,y)$ so $y'' = y' y^j$. Note that $-1<y''<0$ so the point
$(x'',y'')$ is in Regions~B or~G, as required.\end{proof}

\subsection{The Flow Polynomial}
\label{sec:flow}

In order to implement new edge weights from Region~F
(and also to show tractability results and NP-completeness results for Region~F in Section~\ref{Sec:PosE})
we must introduce a specialisation of the Tutte polynomial called the flow polynomial.

A $q$-\emph{flow} of an undirected graph~$G=(V,E)$
is defined as follows~\cite[Section 2.4]{sokal}.
Choose an arbitrary direction for each edge.
Let $H$ be any Abelian group of order~$q$.
A $q$-flow is a mapping $\psi:E\rightarrow H$
such that the flow into each vertex is equal to the flow out (doing
arithmetic in~$H$). 

Consider the following polynomial, 
where the sum is over $q$-flows of~$G$
 (see \cite[(2.21)]{sokal}).
 $$F(G;q, u) = \sum_{\psi} \prod_{e \in E} \left(1 +  u\delta(\psi(e),0)\right),$$ 
where $\delta$ is the Kronecker delta function defined by $\delta(a,b)=1$ if $a=b$
and $\delta(a,b)=0$ otherwise.
It is a non-trivial fact
that $F(G;q,u)$ depends only on $q$, the size of $H$, and not on $H$ itself.
This polynomial is related to the Tutte polynomial
via the following identity \cite[(2.22)]{sokal}.
\begin{fact}\label{fact:Flow}
If $q$ is a positive integer then
$F(G;q, q/\gamma) = q^{-|V|}  {\big(\frac{q}{\gamma}\big)}^{|E|} Z(G;q,\gamma)$.
\end{fact}
The \emph{flow polynomial} of~$G$, which we write as  $F(G;q)$, is given by~$F(G;q,-1)$.
A $q$-flow $\psi$ of a graph~$G=(V,E)$ is said to be \emph{nowhere-zero} if, for every $e\in E$,
$\psi(e)\neq 0$.
From Fact~\ref{fact:Flow} it is easy to see that if $q$ is a positive integer
then $F(G;q) = q^{-|V|}  {\left( -1\right)}^{|E|} Z(G;q,-q)$ is the number of nowhere-zero $q$-flows of~$G$.

\subsection{Region F}

\begin{lemma}
\label{lem:Eq01}
Suppose $(x,y)$ is a point satisfying $0<x<1$ and $y<-1$  and
$0<(x-1)(y-1)<1$.
Then $(x,y)$ can be used to implement a point $(x_1,y_1)$ with $0<y_1<1$.
\end{lemma}
\begin{proof}
Let $j$ be a sufficiently large positive integer such that $x^j<1-q$.
Implement $(x_1,y_1)$ by a $j$-stretch of $(x,y)$ so that $y_1 = q/(x^j-1)+1$.
\end{proof}

\begin{lemma}
\label{lem:Eq12}
Suppose $(x,y)$ is a point satisfying $0<x<1$ and $y<-1$  and
$1<(x-1)(y-1)<2$.
Then $(x,y)$ can be used to implement a point $(x_1,y_1)$ with $-1<y_1<0$.
\end{lemma}
\begin{proof}Let $q=(x-1)(y-1)$.
Note that $1-q/2>0$ since $q<2$.
Let 
$j$ be a sufficiently large integer  that 
$0<x^j < 1-q/2$.
Note that $1-q<0$ so
$1-q < x^j < 1-q/2$.
Implement $(x_1,y_1)$ by $j$-stretch from the point~$(x,y)$ so
that $y_1=q/(x^j-1)+1$. Note that  $-1<y_1<0$.  \end{proof}
 
\begin{lemma}
\label{lem:Eq2}
Suppose $(x,y)$ is a point satisfying $0<x<1$ and $y<-1$  
for which $q=(x-1)(y-1)$ is not an integer.
Suppose $2<q<4$. 
Then $(x,y)$ can be used to implement a point $(x',y')$ with $x'<0$.
\end{lemma} 
\begin{proof}

Let $q=(x-1)(y-1)$.
As in the proof of Lemma~\ref{lem:F}, we
start by examining what points we can implement from the point $(x',y')=(0,1-q)$ and from points
nearby. 

Suppose that $G$ is a graph  
which contains the edge $(s,t)$. Let $\Gamma = G - (s,t)$.
Following the approach of
Lemma~\ref{lem:F},
let 
\begin{equation}
\label{defneww} 
w(q) = 
\frac{q Z_{st}(\Gamma;q, -q)}{Z_{s|t}(\Gamma;q, -q)},
\end{equation}
which is the weight implemented by $\Gamma$ with edge weight~$-q$.

Then, using similar reasoning to the derivation of (\ref{eq:first}),
\begin{align}
Z(G;q,-q) &= Z_{st}(\Gamma;q,-q)(1-q) + \tfrac1q Z_{s|t}(\Gamma;q,-q)(q-q) \nonumber\\
                   &= Z_{st}(\Gamma;q,-q)(1-q) . \label{eq:second}
\end{align}

Also, 
\begin{align*}
Z_{s|t}(\Gamma;q,-q) &= Z(\Gamma;q,-q) - Z_{st}(\Gamma;q,-q)\\
                                       &= Z(\Gamma;q,-q) + Z(G;q,-q) /(q-1).
                                       \end{align*}
                                       
Thus,  we can use Fact~\ref{fact:Flow} to see that
\begin{align*}
w(q) &= \frac{q Z_{st}(\Gamma;q,-q)}
{Z_{s|t}(\Gamma;q,-q)}     \\                  
&= -q \left(\frac{Z(G;q,-q)/(q-1)}
{Z(\Gamma;q,-q) + Z(G;q,-q) /(q-1)}\right)\\
&= -q \left(\frac{Z(G;q,-q) }
{(q-1)Z(\Gamma;q,-q) + Z(G;q,-q)  }\right) \\
&= -q   \left( \frac{F(G;q) }
{ F(G;q) - (q-1)  F(\Gamma;q) }\right)
.
\end{align*}

First, suppose $2<q<3$.
Following the reasoning in Lemma~\ref{lem:F}, we will show below that,
for a suitable $G$,
$F(G;q)>0$ and $F(\Gamma;q)<0$. Together, these imply that the denominator 
$ F(G;q) - (q-1)  F(\Gamma;q)$ is positive and also that it is larger than the numerator~$F(G;q)$.
Thus, $w(q)<0$ and $w(q)>-q$.
(It is our goal to implement a $\gamma'$ in the range $-q<\gamma'<0$ since, 
for this $\gamma'$, $q/\gamma'+1<0$, so the corresponding $x$-coordinate is less than~$0$.)

By continuity, there is a positive~$\epsilon$ (which depends upon~$q$ 
and~$G$) such that,  
if
$|z- Z_{st}(\Gamma;q,-q)|\leq\epsilon$ 
and $|z'- Z_{s|t}(\Gamma;q,-q)|\leq\epsilon$, 
then
  $-q < \frac{q  z}{ z'}
 <0$. As in the proof of Lemma~\ref{lem:F}, we can show that, for a sufficiently small $\delta\in (0,1)$,
  $|Z_{st}(\Gamma;q,-q-\delta)- Z_{st}(\Gamma;q,-q)|\leq\epsilon$ 
and $|Z_{s|t}(\Gamma;q,-q-\delta)- Z_{s|t}(\Gamma;q,-q)|\leq\epsilon$.
Then we finish by implementing the weight~$-q-\delta$ from the given $(x,y)$ using a large stretch
so $$(x'',y'') = (x^k,q/(x^k-1)+1)= (  \delta/(q+\delta),1-q-\delta).$$

For $3<q<4$ the proof will be similar except that we will establish
$F(G;q)<0$ and $F(\Gamma;q)>0$ so that the denominator of the final expression for $w(q)$ is
negative and is larger in absolute value than the numerator.

To complete the proof, we must establish
that $F(G;q)$ and $F(\Gamma;q)$ have different signs.
Let  $G$ be the Petersen graph. 
Since $G$ is edge-transitive, the edge $(s,t)$ may be chosen arbitrarily.
It can be verified, e.g., using Maple, that
$$F(G,q) = q^6 - 15 q^5 + 95 q^4 - 325 q^3 + 624 q^2 - 620q + 240,$$
and
$$F(\Gamma,q) = q^5 - 12 q^4 + 58 q^3 -138 q^2  + 157 q - 66.$$
Now we note that $F(G;q)$ has four real zeroes at $q= 1,2,3,4$
and two complex zeroes, and $F(G;2.5)>0$.
Also, $F(\Gamma;q)$ has three real zeroes at $q=1,2,3$ and two complex zeroes, and $F(\Gamma;2.5)<0$.
\end{proof}

\begin{remark}
\label{rem:flow}
 
The construction used in the proof of Lemma~\ref{lem:Eq2} breaks down for $q>4$ because
$F(G;q)$ and $F(\Gamma;q)$ have  the same signs.
It is conceivable that the 
lemma could be proved for non-integer $q$ in the 
range $4<q<6$  by using a
generalised Petersen graph rather than a Petersen graph in the construction. 
Indeed, Jacobsen and Salas have shown \cite{JacobsenSalas}
that there are generalised Petersen graphs whose flow polynomials have roots between~$5$ and~$6$.
Given the current state of knowledge, we are pessimistic about the
prospects of proving the lemma for all $q>4$. 
Currently, it is an open question~\cite{JacobsenSalas} 
whether there is a  uniform upper bound~$Q$ for real zeros of arbitrary bridgeless graphs
(so that every bridgeless graph~$G$ would have $F(G;q)>0$ for all $q>Q$).
If so, then computing the sign of the flow polynomial will be trivial 
for $q>Q$, so computing the sign of the Tutte polynomial will also be trivial for $y<-Q+1$ along the $y$-axis.
If not, then it seems likely that the hardness construction can be extended.
(Thus, it doesn't seem to be possible to resolve all of the unresolved points in Region~F without solving
the open problem about flow polynomials.)
\end{remark}

\section{The main theorem}
\label{sec:mainthm}

This section is devoted to a formal statement of our results concerning the
complexity of $\signtutte(q,\gamma)$ and $\tutte(q,\gamma)$.  In what follows,
\#P-hardness is defined with respect to 
polynomial-time Turing reductions.
NP-hardness is defined by a many-one reduction from an NP-complete decision problem,
whose instance is a ``yes instance'' if the corresponding instance of $\signtutte(q,\gamma)$
has a positive sign, and a ``no instance'' otherwise.
In Figure~\ref{fig:one}, which is a pictorial representation of  our theorem, 
\#P-hard points are depicted in red,
NP-complete points are depicted in blue, and FP points are depicted in green.
Points depicted in white are unresolved.

Theorem~\ref{thm:main} gives 
a complete description of what we know about the complexity 
of $\signtuttep$ and $\tuttep$.  For consistency with existing work by 
a variety of authors, we classify the complexity in terms of the $(x,y)$
parameterisation.  Throughout, we maintain the connection
between the parametrisations $(q,\gamma)$ and $(x,y)$, so that always $\gamma=y-1$
and $q=(x-1)(y-1)$.

\setcounter{counter:save}{\value{theorem}}
\setcounter{theorem}{0} 
\begin{theorem} \label{thm:main}

The points in the $(x,y)$ plane are classified as follows.

\begin{itemize}
\item Region A: Points $(x,y)$ with $x\geq 0$ and $y\geq 0$.  
In this region, $\signtuttep\in\FP$ and $\tuttep\in\numPQ$. 
When $q=0$ we have the stronger $\tuttep\in\FP$.
 
\item Region B: Points $(x,y)$ with 
$\min(x,y) \leq -1$ and $\max(x,y)<0$.  
In this region $\signtuttep$ is $\numP$-hard, except at the point $(x,y)=(-1,-1)$, where
$\tuttep\in\FP$. 

\item Region C: Points $(x,y)$ with $x<-1$ and $y>1$.  In this region $\signtuttep$ is $\numP$-hard.

\item Region D: Points $(x,y)$ with $x>1$ and $y<-1$.  In this region $\signtuttep$ is $\numP$-hard.
  
\item Region E: Points $(x,y)$ with $x\leq -1$ and $0< y \leq 1$. 
Note that these points have $q\geq 0$.
When $q=0$, we have $\tuttep\in\FP$.
When $q\not=0$ is an integer, $\signtuttep\in\FP$ and $\tuttep\in\numPQ$.
When $q$ is a non-integer, $\signtuttep$ is $\numP$-hard, 
apart from the line segment with $x=-1$
and $11/27\leq y < 1$, which is unresolved.  

\item Region F: Points $(x,y)$ with $0< x \leq 1$ and $y\leq -1$. 
Once again, these points have $q\geq 0$.
When $q=0$, we have $\tuttep\in\FP$.
When $q\not=0$ is an integer, $\signtuttep\in\FP$ and $\tuttep\in\numPQ$.
When $q$ is a non-integer satisfying $0<q<4$, $\signtuttep$ is $\numP$-hard,
apart from the line segment with $y=-1$ and $11/27\leq x<1$, which is unresolved. 
Points with non-integer $q>4$ are also unresolved. 

\item The boundary between regions~B and~E: Points $(x,y)$ with $x\leq -1$ and $y=0$.
Note that $q\geq2$.
When $q$ is not an integer, i.e., $x$ is not an integer, $\signtuttep$ is $\numP$-hard.
At $(x,y)=(-1,0)$ we have $\tuttep\in\FP$, while 
at the rest of the points 
$(x,0)$, where $x$ is a negative integer,
$\signtuttep$ is $\NP$-complete and $\tuttep\in\numPQ$.

\item The boundary between regions~B and~F: Points $(x,y)$ with $x=0$ and $y\leq -1$.
 Note that $q\geq2$.
When $2<q<4$  
is not an integer, i.e., $-3<y<-1$ is not an integer, 
$\signtuttep$ is $\numP$-hard.
When $q>4$ is not an integer, i.e., $y<-3$ is not an integer,
the complexity of $\signtuttep$ and $\tuttep$ is unresolved.
At the points $(0,-2)$ and $(0,-3)$, $\signtuttep$ is $\NP$-complete and $\tuttep\in\numPQ$.
The complexity at the point $(0,-4)$ is unresolved.
At the rest of the points $(0,y)$, where $y\leq-5$ is a negative integer, 
$\signtuttep\in\FP$ and $\tuttep\in\numPQ$.

\item Region G: Points $(x,y)$ with $\max(|x|,|y|)<1$ and $q>32/27$.
In this region, $\signtuttep$ is $\numP$-hard. 

\item Region H: Points $(x,y)$ with $\max(|x|,|y|)<1$, $q\leq 32/27$ and $x<-2y-1$.
In this region, $\signtuttep$ is $\numP$-hard, apart from points with 
$q=1$, where $\tuttep\in\FP$. 
 
\item Region I: Points $(x,y)$ with $\max(|x|,|y|)<1$, $q\leq 32/27$ and $y<-2x-1$.
In this region, $\signtuttep$ is $\numP$-hard, apart from points with 
$q=1$,  where  $\tuttep\in\FP$.

\item Region J: Points $(x,y)$ with $-1\leq x < 0$ and $y\geq 1$. 
In this region, $\signtuttep\in\FP$ and $\tuttep\in\numPQ$.

\item Region K: Points $(x,y)$ with $x\geq 1$ and $-1\leq y < 0$. 
In this region, $\signtuttep\in\FP$ and $\tuttep\in\numPQ$.

\item Region L: Points $(x,y)$ with $0<x<1$ and $-x<y<0$.
In this region, $\signtuttep\in\FP$ and $\tuttep\in\numPQ$.

\item Region M: Points $(x,y)$ with $0<y<1$ and $-y<x<0$.
In this region, $\signtuttep\in\FP$ and $\tuttep\in\numPQ$.

\item The rest: There are some remaining unresolved points. These points (simultaneously) satisfy
all of the following inequalities:  $\max(|x|,|y|)<1$, $y<-x$, 
$q\leq 32/27$, $y\geq -2x-1$, $x\geq -2y-1$,     
and $q\neq 1$.
\end{itemize}

\end{theorem}
\setcounter{theorem}{\value{counter:save}}

\begin{proof}

The proof follows from the following lemmas, which appear in the rest of the paper.

\begin{itemize}
\item Region A:  Lemma~\ref{lem:upper}.

\item Region B:  Corollaries \ref{cor:B}--\ref{cor:BGhorizontalboundary} and Section~\ref{refsec}.

\item Region C:  Corollary~\ref{cor:CD}.

\item Region D:  Corollary~\ref{cor:CD}.
 
\item Region E:  Corollaries \ref{cor:F} and \ref{cor:GFboundary}, and Observation~\ref{obs:Fpos}. 

\item Region F:  Corollaries \ref{cor:E} and \ref{cor:GEboundary}, and Observation~\ref{obs:Epos}.

\item The boundary between regions~B and~E:  
Corollary \ref{cor:FBboundary} and Observation~\ref{obs:posBF}.

\item The boundary between regions~B and~F:  
Corollary \ref{cor:EBboundary} and Observation~\ref{obs:flows}.

\item Region G:  Corollary~\ref{cor:GIH}.

\item Region H:  Corollaries \ref{cor:GIH} and \ref{cor:IH}, and Section~\ref{sec:HI}.
 
\item Region I: Corollaries \ref{cor:GIH} and \ref{cor:IH}, and Section~\ref{sec:HI}.

\item Region J:  Corollary~\ref{cor:J}.

\item Region K:  Corollary~\ref{cor:K}.

\item Region L:  Corollary~\ref{cor:L}.

\item Region M:  Corollary~\ref{cor:M}.

\end{itemize}

\end{proof}

All \#P-hardness results are proved in the following Section~\ref{sec:red}. 
Tractability results and NP-completeness results
are proved in Section~\ref{sec:positive} where we also 
show that $\tutte(q,\gamma)$ is in $\numPQ$ for these points.

\section{\#P-hardness}
\label{sec:red} 
In this section and the next we use the following shorthand.
We say that a point $(x,y)$ is \#P-hard, NP-complete, or in FP,
if, for $\gamma=y-1$ and $q=(x-1)(y-1)$, the corresponding problem 
$\signtutte(q,\gamma)$ is \#P-hard, NP-complete, or in FP, respectively.

\subsection{Points in Region B}

\begin{corollary}
\label{cor:B} Suppose that $(x,y)$ is a point 
such that $\min(x,y) < -1$ and $\max(x,y)<0$. 
Then~$(x,y)$ is $\numP$-hard.
\end{corollary} 
 
\begin{proof}
Note that $q=(x-1)(y-1)>1$. The corollary follows from Lemmas~\ref{lem:signhardqbig}, \ref{lem:shiftB1},
 \ref{lem:shiftB2}, \ref{lem:shiftB3}, and \ref{lem:shiftB4}.\end{proof} 
 
\begin{corollary}
\label{cor:BGverticalboundary}
Suppose that $(x,y)$ is a point satisfying $x=-1$ and $-1<y<0$.
Then $(x,y)$ is $\numP$-hard. 
\end{corollary} 
\begin{proof}
A $3$-thickening from $(x,y)$ implements the point 
$$(x',y') = \left(\frac{-1+y+y^2}{1+y+y^2},y^3\right).$$
Now $x'<-1$ and $-1<y'<0$ so $(x',y')$ was
already shown to be \#P-hard by Corollary~\ref{cor:B}.
\end{proof}

Similarly, we have the following.
\begin{corollary}
\label{cor:BGhorizontalboundary}
Suppose that $(x,y)$ is a point satisfying $y=-1$ and $-1<x<0$.
Then $(x,y)$ is $\numP$-hard. 
\end{corollary}  
 
\subsection{Points in Regions~G, H and I }

\begin{corollary}
\label{cor:GIH} Suppose that $(x,y)$ is a point 
satisfying $\max(|x|,|y|)<1$
  and $q=(x-1)(y-1)>1$.
Suppose that $(x,y)$ also satisfies at least one of the following conditions.
\begin{itemize}
\item $q>32/27$, or
\item $y<-1-2x$, or
\item $x<-1-2y$. 
\end{itemize}
Then~$(x,y)$ is $\numP$-hard.
\end{corollary} 
\begin{proof}
The corollary follows from Lemmas~\ref{lem:signhardqbig}, \ref{lem:3227g2}, \ref{lem:triangle1}, \ref{lem:triangle2}, 
and \ref{lem:3227g1}.
\end{proof}

\begin{corollary}
\label{cor:IH} Suppose that $(x,y)$ is a point 
satisfying $\max(|x|,|y|)<1$
and $q=(x-1)(y-1)<1$.
Suppose that $(x,y)$ also satisfies at least one of the following conditions.
\begin{itemize}
\item $y<-1-2x$, or
\item $x<-1-2y$. 
\end{itemize}
Then~$(x,y)$ is $\numP$-hard.
\end{corollary} 
\begin{proof}
Note that $q>0$.
The corollary follows from Lemmas~\ref{lem:signhardqsmall}, 
\ref{lem:triangle1} and \ref{lem:triangle2}.
We implement the point $(x_1,y_1)$ required by Lemma~\ref{lem:signhardqsmall}
by taking a $2$-thickening of~$(x,y)$ so $y_1 = y^2 \in (0,1)$.
\end{proof}

\subsection{Points in Regions~C and~D}

\begin{corollary}
\label{cor:CD}
Suppose $(x,y)$ is a point satisfying  
one of the following.
\begin{itemize}
\item $y>1$ and $x<-1$, or
\item $x>1$ and $y<-1$.
\end{itemize}
Then~$(x,y)$ is $\numP$-hard. 
\end{corollary}
\begin{proof}
Note that $q<0$.
The corollary follows from Lemmas~\ref{lem:signhardqsmall}
and~\ref{lem:stretchCD}. The point $(x_2,y_2)$ required by Lemma~\ref{lem:signhardqsmall}
is just $(x,y)$ itself.
\end{proof}
 
\subsection{Points 
with non-integer $q$ in Region E and on the boundary between Regions B and~E}
\label{sec:F} 
  
Note that $q$ is an integer when $(x,y)=(-1,0)$ and
when  $y=1$. We will discuss these points in Section~\ref{sec:positive}.
  
\begin{corollary}
\label{cor:F}
Suppose $(x,y)$ is a point satisfying $x<-1$ and $0<y<1$.
Suppose that $q=(x-1)(y-1)>0$ is not an integer.
Then $(x,y)$  is $\numP$-hard.
\end{corollary}

\begin{proof}
If $0<q<1$ then the result follows from Lemmas~\ref{lem:signhardqsmall} and \ref{lem:xlefttoyup}.
If $1<q<2$ then the result follows from Lemmas~\ref{lem:signhardqbig},
\ref{lem:Fq12} and  \ref{lem:xlefttoyup}.
So suppose $q>2$.
By Lemma~\ref{lem2:F}, the point $(x,y)$ can be used to implement 
a point, other than the special point $(-1,-1)$  that is in  Regions B or~G.
All of these points are known to be \#P-hard by
Corollaries~\ref{cor:B}, \ref{cor:BGverticalboundary}, \ref{cor:BGhorizontalboundary}
and~\ref{cor:GIH}.
\end{proof}

\begin{corollary}
\label{cor:FBboundary}
Consider a point $(x,y)$ satisfying $x<-1$
and $y=0$.
Suppose that $q=(x-1)(y-1)$ is not an integer.
Then $(x,y)$ is
$\numP$-hard.
\end{corollary}
\begin{proof}
Note that $q=(x-1)(0-1)=1-x>0$. 
Let 
$$(x',y') = \left(x^3,\frac{x+x^2}{1+x+x^2}\right)$$
be the point implemented by a $3$-stretch from $(x,y)$.
Note that $x+x^2>0$ so $0<y'<1$. Also, $x'<-1$.
Thus, $(x',y')$ is \#P-hard by Corollary~\ref{cor:F}.
\end{proof}

\begin{corollary}
\label{cor:GFboundary}
Suppose that $(x,y)$ is a point satisfying $x=-1$ and $0<y<11/27$.
Then $(x,y)$ is $\numP$-hard.
\end{corollary}
\begin{proof}
Note that $q=(x-1)(y-1)>32/27$.
Implement $(x',y')$ by a $2$-thickening from $(x,y)$ so
$(x',y') = \big(\frac{-1+y}{1+y},y^2\big)$. Note that $-1<x'<0$ and $0<y'<1$ so
$(x',y')$ is in Region~G, and is \#P-hard by Corollary~\ref{cor:GIH}.
\end{proof}

\subsection{Points with non-integer $q$  in Region F and on the boundary between regions B and~F}
\label{sec:E}

Note that $q$ is an integer when $(x,y)=(0,-1)$ and when $x=1$. We will discuss these points in 
Section~\ref{sec:positive}.

 \begin{corollary}
\label{cor:E}
Suppose $(x,y)$ is a point satisfying $0<x<1$ and $y<-1$.
Suppose that $q=(x-1)(y-1)$ is not an integer.
Suppose $0<q<4$.
Then $(x,y)$  is $\numP$-hard.
\end{corollary}

\begin{proof} 
If $0<q<1$ then the result follows from Lemmas~\ref{lem:signhardqsmall} and \ref{lem:Eq01}.
If $1<q<2$ then the result follows from Lemmas~\ref{lem:signhardqbig} and
\ref{lem:Eq12}.
So suppose $2<q<4$.
By Lemma~\ref{lem:Eq2}, $(x,y)$ can be used to implement a point
$(x',y')$ with $x'<0$. The point $(x',y')$ is in one of regions E, B, or G. 
It is not the special point $(-1,-1)$ from Region~B, since $q$ is not an integer.
It is not the unresolved line segment from Region~E, since $q>2$.
Thus, $(x',y')$ is \#P-hard by Corollaries~\ref{cor:B}, \ref{cor:BGverticalboundary}, \ref{cor:BGhorizontalboundary},
\ref{cor:GIH}, \ref{cor:F}, \ref{cor:FBboundary},
and \ref{cor:GFboundary}.
\end{proof} 
 
As we explained in Remark~\ref{rem:flow}, it seems possible that Corollary~\ref{cor:E} could be extended, say
up to $q=6$, by doing more complicated calculations in the proof of Lemma~\ref{lem:Eq2}, analysing the flow
polynomial of generalised Petersen graphs, rather than just the flow polynomial of the Petersen graph.
However, our lack of knowledge about the zeroes of the flow polynomial seems to be a barrier
to extending the lemma to cover all~$q$.

 \begin{corollary}
\label{cor:EBboundary}
Consider a point $(x,y)$ satisfying $x=0$ and $y<-1$.
Suppose that $q=(x-1)(y-1)$ is not an integer and that $q<4$.
Then $(x,y)$ is $\numP$-hard.
\end{corollary}
\begin{proof} 
Note that  $2<q<4$.
Let 
$$(x',y') = \left(\frac{y+y^2}{1+y+y^2},y^3\right)$$
be the point implemented by a $3$-thickening from $(x,y)$.
Note that $y+y^2>0$ so $0<x'<1$. Also, $y'<-1$.
Thus, $(x',y')$ is \#P-hard by Corollary~\ref{cor:E}.
\end{proof}
 
 \begin{corollary}
\label{cor:GEboundary}
Suppose that $(x,y)$ is a point satisfying $0<x<11/27$ and $y=-1$.
Then $(x,y)$ is $\numP$-hard.
\end{corollary}
\begin{proof}  
Note that $q=(x-1)(y-1)>32/27$.
Implement $(x',y')$ by a $2$-stretch from $(x,y)$ so
$(x',y') = \big(x^2,\frac{-1+x}{1+x}\big)$. Note that $0<x'<1$ and $-1<y'<0$   so
$(x',y')$ is in Region~G, and is \#P-hard by Corollary~\ref{cor:GIH}.
\end{proof}

\section{Tractability results and NP-completeness results}  
\label{sec:positive}
 
As we mentioned earlier,  
we say that a point $(x,y)$ is in FP
if $\signtutte(q,\gamma)$
can be solved in polynomial time,
where $q=(x-1)(y-1)$ and $\gamma=y-1$.
These points are depicted in green in Figure~\ref{fig:one}.
For each point in FP, and also for the points that are NP-complete (depicted in blue), we show that
$\tutte(q,\gamma)$ is  in $\numPQ$. Thus,
$\tutte(q,\gamma)$ can be efficiently approximated using an NP oracle.  
 
\subsection{Points in Region A}

The following lemma is implicit in the work of Tutte~\cite{Tutte54,Tutte84}.
The connection is explained explicitly in~\cite[Section 2.3]{tutteone}.

\begin{lemma}
\label{lem:upper} 
Suppose $(x,y)$ is a point satisfying $\min(x,y)\geq 0$. 
Let $q=(x-1)(y-1)$ and $\gamma=y-1$.
Then for every graph~$G$,
$Z(G;q,\gamma)>0$ so $\signtutte(q,\gamma)$ is in $\FP$.
Furthermore,  $\tutte(q,\gamma)$ is in $\numPQ$. 
In the case $q=0$, we have $Z(G;q,\gamma)=0$ and $\tutte(q,\gamma)$ is trivially
in $\FP$.
\end{lemma}

\subsection{Points in Region B}
\label{refsec}

It is known \cite{JVW90} that 
$\tutte(4,-2)$ is in FP (so it is certainly in $\numPQ$). 
Thus, the point $(x,y)=(-1,-1)$ is in FP.
     
\subsection{Points with Integer $q$ in Region E}
\label{Sec:PosF}

The points  in Region~E have $x\leq -1$ and $0< y \leq 1$.
Thus, they have $q=(x-1)(y-1)\geq0$ and $\gamma=y-1$.  

First, if $y=1$ then $q=0$. We will handle this easy case below.  So, suppose $y<1$ so  $-1 < \gamma<0$.
Note that $q> 0$ so, since we restrict attention to integer~$q$, $q\geq 1$.  
Consider the Potts-model partition function for $G$ (see \cite[(2.7)]{sokal}).

$$Z_{\mathrm{Potts}}(G;q,\gamma) = \sum_{\sigma: V \rightarrow [q]} \prod_{e=(u,v) \in E}
\left(1 + \gamma \delta(\sigma(u),\sigma(v))
\right),$$
where $\delta$ is the Kronecker delta function defined by $\delta(a,b)=1$ if $a=b$
and $\delta(a,b)=0$ otherwise.
The following well-known fact is due to Fortuin and Kasteleyn (see~\cite[Theorem 2.3]{sokal}).
\begin{fact} \label{fact:Potts}
If $q\geq 1$ is an integer then
$Z_{\mathrm{Potts}}(G;q,\gamma) = Z(G;q,\gamma)$.
\end{fact}

The following observation now follows from Fact~\ref{fact:Potts}.

\begin{observation}\label{obs:Fpos}
Let $(x,y)$ be a point with  $x\leq -1$ and $0< y \leq 1$.
Let $q=(x-1)(y-1)$ and $\gamma=y-1$.
Suppose that $q$ is an integer.
\begin{itemize}
\item
If $y=1$ then $Z(G;q,\gamma)=0$
so $\signtutte(q,\gamma)$ and $\tutte(q,\gamma)$ are both in $\FP$.\footnote{The case $y=1$ is trivial for us, 
because we are using the $(q,\gamma)$ parameterisation, where a single point $(q,\gamma)=(0,0)$ corresponds to the line $(x,1)$ in the $(x,y)$ parameterisation. This issue is touched on in the Introduction.}

\item
Otherwise, 
$Z(G;q,\gamma)> 0$ so $\signtutte(q,\gamma)$ is in $\FP$.
Also, $\tutte(q,\gamma)$ is in $\numPQ$.
\end{itemize}
\end{observation}
 
Note that Observation~\ref{obs:Fpos} disproves \cite[Conjecture 10.3(e)]{JacksonSokal}. 
Jackson and Sokal 
conjectured that for every fixed $x\leq -1$ and $0<y<1$ satisfying 
$q=(x-1)(y-1)>32/27$,  for all sufficiently large~$n$ an~$m$, there are $2$-connected graphs
with $n$ vertices and $m$ edges that make $Z(G;q,y-1)$ non-zero with either sign, but this is clearly false
when $q$ is an integer.   
 
 \subsection{Points with Integer $q$ on the boundary between Regions B and~E}

These points have $x\leq -1$ and $y=0$. Since $q=(x-1)(y-1)=1-x$ is an integer, we conclude
that $x$ is an integer.
From Fact~\ref{fact:Potts},
$Z(G;q,-1)$ is the number of proper $q$-colourings of~$G$.

\begin{observation}\label{obs:posBF}
The point $(-1,0)$ is in $\FP$ since $Z(G;2,-1)$ is equal to the number of $2$-colourings of~$G$,
and this can be computed in polynomial time.
For integer $x<-1$, 
the point $(x,0)$ is NP-complete. $Z(G;1-x,-1)$ is positive if $G$ has a proper $(1-x)$-colouring,
and is $0$ otherwise. $\tutte(1-x,-1)$ is in $\numP$ so it is in $\numPQ$.
\end{observation}

\subsection{Points with Integer $q$ in Region F}
\label{Sec:PosE}
 
The points in Region~F have  $0< x \leq 1$ and $y\leq -1$.
They have $q=(x-1)(y-1)\geq 0$ and $\gamma=y-1$.

First, if $x=1$ then $q=0$. We will handle this easy case below.
So, let us restrict attention to the range $0\leq x < 1$.
This corresponds to $\gamma\leq-2$ and $q/\gamma\in (-1,0)$. 
Recall the definition of the flow polynomial from Section~\ref{sec:flow}.
Using Fact~\ref{fact:Flow} we obtain the following observation.

\begin{observation} 
\label{obs:Epos}
Let $(x,y)$ be a point with  $0< x \leq 1$ and $y\leq -1$.
Let $q=(x-1)(y-1)$ and $\gamma=y-1$.
Suppose that $q$ is an integer.
\begin{itemize}
\item
If $x=1$ then $Z(G;q,\gamma)=0$
so $\signtutte(q,\gamma)$ and $\tutte(q,\gamma)$ are both in $\FP$.
\item
Otherwise, 
$q^{-|V|}  {\left(\frac{q}{\gamma}\right)}^{|E|}
Z(G;q,\gamma)> 0$ so $\signtutte(q,\gamma)$ is in $\FP$.
Also, $\tutte(q,\gamma)$ is in $\numPQ$.
\end{itemize}
\end{observation}
 
Like Observation~\ref{obs:Fpos},
Observation~\ref{obs:Epos} provides counter-examples to \cite[Conjecture 10.3(3)]{JacksonSokal}.
  They conjectured that for every fixed $0<x\leq 1$ and $y\leq -1$ satisfying
$q=(x-1)(y-1)>32/27$,  for all sufficiently large~$n$ and~$m$ (including even~$m$), there are $2$-connected graphs
with $n$ vertices and $m$ edges that make $Z(G;q,y-1)$ non-zero with either sign, but this is clearly false
when $q$ is an integer.
 
\subsection{Points with Integer $q$ on the boundary between Regions B and~F}
 
These points have $x=0$ and $y\leq -1$. Since $q=(x-1)(y-1)=1-y$ is an integer, we conclude that $y$
is an integer.

Recall from Section~\ref{sec:flow} that
if $q$ is a positive integer
then $q^{-|V|}  {\left( -1\right)}^{|E|} Z(G;q,-q)$ is the number of nowhere-zero $q$-flows of~$G$.
A graph has a nowhere-zero $2$-flow iff it is Eulerian~\cite[Theorem 11.21]{CLZ}. Thus, this can be
tested in polynomial time.
On the other hand, it is NP-complete to test whether a graph has a nowhere-zero $3$-flow, even if the graph is
planar. To see this, note that
a planar graph has a nowhere-zero $3$-flow iff its dual has a proper $3$-colouring,
and it is NP-complete to determine whether a planar graph is $3$-colourable.
It is also NP-complete to test whether a graph has a nowhere-zero $4$-flow, even if the graph is
cubic. To see this, consider a cubic graph~$G$ and let~$H$ be the Abelian group~$Z_2\times Z_2$.
A $4$-flow maps the edges of~$G$ to~$(0,1)$, $(1,0)$ and~$(1,1)$. To be nowhere-zero, it maps
one of each to the edges adjacent to each vertex. So the number of nowhere-zero $4$-flows is the same
as the number of proper $3$-edge-colourings of~$G$. 
But it is NP-complete to decide whether a graph has such an edge colouring~\cite{Holyer}.
A ``bridge'' (or cut-edge) of a graph is an edge whose deletion increases the number
of connected components.
It is known~\cite[Corollary 11.26]{CLZ}, that no graph with a bridge 
has a nowhere-zero $q$-flow for any integer $q\geq 2$.
However, Seymour has shown
\cite[Theorem 11.32]{CLZ} that every bridgeless graph has a nowhere-zero $6$-flow.
Thus, determining whether a graph has a nowhere-zero $q$-flow is in FP for $q\geq 6$.
We do not know the complexity of determining whether a graph has a nowhere-zero  
$5$-flow. Indeed, it is currently an open question whether 
there exists a bridgeless graph without a nowhere zero $5$-flow.

\begin{observation}\label{obs:flows}
The point $(0,-1)$ is in $\FP$ since $Z(G;2,-2)$ is computable from the number of nowhere-zero $2$-flows of~$G$,
and this can be computed in polynomial time.
The point $(0,-2)$ is $\NP$-complete since $Z(G;3,-3)$ allows one to determine
the number of nowhere-zero $3$-flows of~$G$.
The point $(0,-3)$ is $\NP$-complete since $Z(G;4,-4)$ allows one to determine the
number of nowhere-zero $4$-flows of~$G$.
For integer $y\leq -5$, the point $(0,y)$ is in FP since
$Z(G;1-y,y-1)$ is computable from the number of nowhere-zero $(1-y)$-flows of~$G$.
This quantity is positive iff $G$ has no bridge.
$\tutte(1-x,-1)$ is in $\numP$ so it is in $\numPQ$.
\end{observation}

\subsection{Points in Regions H and I}\label{sec:HI}

It is known \cite{JVW90} that 
points $(x,y)$ with $(x-1)(y-1)=1$ are
in FP since $\tutte(1,\gamma)$ is in FP
so $\signtutte(1,\gamma)$ is also in FP.
   
\subsection{Matroids}
    
The definitions from Section~\ref{sec:preliminaries} can be generalised from graphs to matroids.
To deal with Regions~J and~K (and also with regions~L and~M in 
future sections), it is advantageous 
to work with matroids, rather than with graphs, because we can then exploit 
a duality between the variables~$x$ and~$y$.
In order to avoid difficulties over how matroids should be presented, we will work with the class of binary matroids.
This is a more general class than the class of graphs --- every graph corresponds to a binary matroid, but there are binary matroids 
that do not correspond to graphical matroids.

A matroid $\calM$ is a combinatorial structure  defined  by a 
set $E$ (the ``ground set'') 
together with a ``rank function'' $r_{\calM}:2^E\to\mathbb{N}$
which must satisfy  the following conditions (see \cite{oxley} for details).
\begin{enumerate}
\item $0\leq r_{\calM}(A) \leq |A|$,
\item $A \subseteq B$ implies $r_{\calM}(A) \leq r_{\calM}(B)$ (monotonicity), and
\item $r_{\calM}(A\cup B) + r_{\calM}(A \cap B) \leq r_{\calM}(A) + r_{\calM}(B)$ (submodularity).
\end{enumerate}

A subset $A\subseteq E$ satisfying $r_{\calM}(A)=|A|$
is said to be {\it independent}.
Every other subset $A\subseteq E$ is said to be {\it dependent.}
A  maximal (with respect to inclusion)
independent set is   a {\it basis}, and a minimal dependent set is   a {\it circuit}.
A~circuit with one element is  a {\it loop}. 
 
The
multivariate Tutte polynomial  
of a matroid~$\matroid$ with ground set~$\columns$ and rank function $\rank_{\matroid}$
is 
defined as follows (see~\cite[(1.3)]{sokal}), where the weight function $\bgamma$ 
assigns weights to elements of the ground set.
\begin{equation}
\label{tildedef}
\ZtildeTutte(\matroid;q,\boldgamma)=\sum_{\subsetcols\subseteq\columns}
q^{-\rank_{\matroid}(\subsetcols)}
\prod_{\column\in\subsetcols}
\gamma_\column.
\end{equation}
If $\bgamma$ assigns weight~$\gamma$ to every element of~$E$ then
we use $\ZtildeTutte(\matroid;q,\gamma)$ as shorthand for $\ZtildeTutte(\matroid;q,\bgamma)$.

Let $\matrix$ be a matrix over a field $\field$ with 
row set~$\rows$ and column set~$\columns$.
$\matrix$ is said to ``represent'' a matroid~$\matroid$ with ground set $\columns$.
The rank  $\rank_{\matroid}(\subsetcols)$ of a set of columns~$\subsetcols$
in this matroid
is defined to be the rank of the submatrix consisting of those columns.  A matroid is said to be \emph{representable} over the field~$\field$ if it can be represented in this way.
It is said to be \emph{binary} it is representable over the two-element field~$\gf2$.   
  
The \emph{cycle matroid} of an undirected graph $\graph=(\graphvertices,\graphedges)$
is the binary matroid $\matroid(\graph)$ represented by the vertex-edge
incidence matrix~$\matrix$ of~$\graph$ (in which rows are vertices and columns
are edges).   It can be deduced from the definition above that
$\rank_{\matroid(\graph)}(\subsetcols) = |\graphvertices| - \kappa(\graphvertices,\subsetcols)$.
The Tutte polynomial of a cycle matroid $\matroid(\graph)$ is very closely connected to the Tutte polynomial
of the underlying graph $\graph$ . In particular, (see~\cite[(1.2) and (1.3)]{sokal}),
\begin{equation} 
 \label{eq:graphmatroid}
Z(\graph;q,\boldgamma)= q^{|V|}\,\ZtildeTutte(\matroid(\graph);q,\boldgamma).
\end{equation}

Every matroid~$\calM$ has a dual matroid~$\calM^*$ with the same ground set. Furthermore, $\calM^*$ is
binary if and only if $\calM$ is (see \cite{oxley}), and a binary matrix representing $\calM^*$
can be efficiently computed from a representation of $\calM$~\cite[p.63]{TruemperBook}.
A {\it cocircuit\/} in $\calM$ is a set that is a circuit in~$\calM^*$;
equivalently, a cocircuit is a minimal set that intersects every basis.
A~cocircuit with one element is a {\it coloop}. 
We use the following fact~\cite[(4.14a)]{sokal}.

\begin{fact}
\label{fact:duality}
Suppose that $\calM$ is a matroid with ground set~$E$
and that $\bgamma$ is a weight function assigning weights to elements in~$E$.
Let $\calM^*$ be the \emph{dual} of~$\calM$ and let $\bgamma^*$
be the weight function that assigns weight $q/\gamma_e$ to every ground set element $e\in E$.
Then 
$$\Ztilde(\calM^*;q,\bgamma) = q^{-r_{\calM^*}(E)} \left(
\prod_{e \in E}{\gamma_e}
\right)
\Ztilde(\calM;q,\bgamma^*).$$
\end{fact}

Two important matroid operations are deletion and contraction.  Suppose 
$e \in E $ is a member of the ground set of matroid~$\calM$.  
The {\it contraction
$\calM/e$ of $e$ from $\calM$} is the matroid on ground set $E-\{e\}$ with 
rank function given by $r_{\calM/e}(A)=r_\calM(A\cup \{e\})-r_\calM(\{e\})$, for
all $A\subseteq E-\{e\}$.
The {\it deletion
$\calM\backslash e$ of~$\{e\}$ from $\calM$} is the matroid on ground set $E-\{e\}$ with 
rank function given by $r_{\calM \backslash e}(A)=r_\calM(A)$, for
all $A\subseteq E-\{e\}$.   
Given a matrix representing a matroid~$\calM$, there are efficient algorithms
for constructing matrices representing contractions and deletions of~$\calM$
\cite[Chapter 3]{TruemperBook}.
We use the following fact (see, for example, \cite[(4.18b)]{sokal}).
\begin{fact}
\label{fact:loop}
If $\calM$ is a matroid with a loop~$e$  then 
$$\Ztilde(\calM;q,\bgamma) = (1+\gamma_e) \Ztilde(\calM \setminus e;q,\bgamma).$$
\end{fact} 
We also use a related fact about  coloops (see, for example \cite[(2.6)]{mani}. \begin{fact}
\label{fact:bridge}
If $\calM$ is a matroid with a coloop~$e$ then
$$\Ztilde(\calM;q,\bgamma) = (1+\gamma_e/q) \Ztilde(\calM \contract e;q,\bgamma).$$
\end{fact}

We introduce two computational problems for binary matroids.

\prob{$\msigntutte(q,\gamma)$.}
{A matrix representing a binary matroid~$\calM$
and an edge weight~$\gamma$.}  
{Determine whether the
sign of $\ZtildeTutte(\calM;q,\gamma)$ is positive, negative, or $0$.}

\prob{$\mtutte(q,\gamma)$.}
{A matrix representing a binary matroid~$\calM$ and an  edge weight~$\gamma$.} 
{$\ZtildeTutte(\calM;q,\gamma)$.}

\subsection{Points in Regions J and K}

The points in Regions~J and K
satisfy $-1 \leq \min(x,y) < 0$ and $\max(x,y) \geq 1$.
Let $q=(x-1)(y-1)$ and $\gamma=y-1$. Note that $q\leq 0$.
It is known (see \cite[Theorem 4.1]{JacksonSokal}
that in these regions, the sign of
$Z(G;q,\gamma)$ is essentially a trivial function of~$G$,
apart from some factors arising from loops in the 
matroid associated with~$G$ and in its dual matroid.
We will show that, for all of these points,
$\tutte(q,\gamma)$ is in $\numPQ$.
In fact, we will show that 
$\mtutte(q,\gamma)$ is in $\numPQ$.
Working with matroids, instead of with graphs, will enable us to 
prove the results for one region (Region~K) and immediately to
deduce the same results for the other region (Region~J), 
by duality of the variables~$x$ and~$y$.  (The replacement of $\gamma_e$
with $q/\gamma_e$ in Fact~\ref{fact:duality} is equivalent to swapping 
$x$ and~$y$.)

\subsubsection{Points in Region K}  
 
Points in Region~K have  $x\geq 1$ and $-1\leq y < 0$.
Let $q=(x-1)(y-1)$ and $\gamma=y-1$.

First, if $x=1$ then $q=0$. We will handle this easy case below. 
So, let us restrict attention to the range $x>1$.
Then  $q < 0$ and
$-2 \leq \gamma<-1$. 
We will use the following lemma, which is similar in spirit to  
\cite[Theorem 4.1]{JacksonSokal}.\footnote{We
need to repeat the steps of their proof here because we want to extract computational information
in addition to the sign.}

\begin{lemma}\label{lem:matroid} Suppose that $q<0$ and $\calM$ is a loopless
matroid. Suppose that $\bgamma$ is a weight function in
which every weight $\gamma_e$ satisfies  $-2 \leq \gamma_e \leq 0$.  
Then $\Ztilde(\calM;q,\bgamma)> 0$ and 
the problem of computing
$\Ztilde(\calM;q,\bgamma)$ is in $\numPQ$.
\end{lemma}

\begin{proof} 

We start with some pre-processing.
Before trying to compute $\Ztilde(\calM;q,\bgamma)$, we first
modify $\calM$, without changing its Tutte polynomial,
to get rid of any size-$2$ circuits.
We do this by parallel composition.
So if we have a size-$2$ circuit containing elements $e_1$ and $e_2$, 
we replace it 
with a new element~$e$ which is the parallel composition of the two
elements in the circuit.
In the matrix representing~$\calM$, the size-$2$ circuit
arises as a pair of identical columns. In the representation of
the new matroid, 
the columns corresponding to elements~$e_1$ and~$e_2$ are deleted and
the new element~$e$ corresponds to one of these columns.
The new weight $\gamma_e$ is given by 
$\gamma_{e_1} + \gamma_{e_2} + \gamma_{e_1} \gamma_{e_2}$
(see \cite[2.34]{JacksonSokal}).
The reason that we want to do this pre-processing is that, in the recursive step, we will want to be able
to contract an element of a circuit without  creating a loop.
The reason that we can do
the pre-processing 
without falsifying the conditions in the statement of the lemma
is that the region $-2 \leq \gamma \leq 0$ maintains itself for parallel composition:
If $-2\leq \gamma_{e_1} \leq 0$ and $-2\leq \gamma_{e_2} \leq 0$
then  $-2 \leq \gamma_e \leq 0$.

Now suppose that $\calM$ has no size-$2$ circuit. Let $r = r_\calM$ and
$E = E(\calM)$. Then

$$\Ztilde(\calM;q,\bgamma)=\sum_{A\subseteq E} q^{-r(A)} \prod_{e \in A} \gamma_e.$$
 .

{\bf Base Case:\quad} If $r(E)=|E|$ then, from the axioms of rank functions of
matroids, for every $S\subseteq E$, 
$r(S)=|S|$, so
$$\Ztilde(\calM;q,\bgamma)=\sum_{A\subseteq E} q^{-|A|} \prod_{e \in A} \gamma_e
= \sum_{A\subseteq E} 
\prod_{e \in A} \frac{\gamma_e}{q}.$$
The contribution from $A=\emptyset$ is~$1$ and the contribution from each other~$A$ is non-negative.
Also, $\Ztilde(\calM;q,\bgamma)$ can be computed by summing over the sets~$A$.

{\bf Recursive Step:\quad}
Pick any $e$ in a circuit.
Then from \cite[(4.18a)]{sokal},
$$\Ztilde(\calM;q,\bgamma) = 
\Ztilde(\calM \setminus e;q,\bgamma) + \frac{\gamma_e}{q}
 \Ztilde(\calM / e;q,\bgamma).$$
 
 Now the point is that the fraction $\gamma_e/q$ doesn't change the sign, and is easy
 to compute. Also,
 the two minors $\calM \setminus e$
 and $\calM / e$
 both satisfy the conditions of the theorem.
 
 Both minors are matroids on ground set $E \setminus e$.
 The rank functions are given by 
 $r_{\calM \setminus e}(A) = r(A)$
 and $r_{\calM /e}(A) = r(A\cup e)-1$.
 
 To see that $\calM/e$ has no loop, note
 that $r_{\calM/e}(\{e'\}) = r(\{e,e'\})-1$ and since $\{e,e'\}$ is not a circuit, by the
 pre-processing step, $r(\{e,e'\})=2$.
 \end{proof}

We can now classify the points in Region~K. 
See also \cite[Theorem 4.1]{JacksonSokal} which shows that the sign is trivial in this region.

\begin{lemma}\label{lem:K}
Let $(x,y)$ be a point with $x\geq 1$ and $-1\leq y < 0$. 
Let $q=(x-1)(y-1)$ and $\gamma=y-1$.
Then   $\msigntutte(q,\gamma)$ is in $\FP$ and $\mtutte(q,\gamma)$ is in $\numPQ$.
\end{lemma}

\begin{proof}
If $\calM$ has $k$ loops then, by 
Fact~\ref{fact:loop}, 
$\ZtildeTutte(\calM;q,\gamma) ={ (1+\gamma)}^k \ZtildeTutte(\calM';q,\gamma)$,
where $\calM'$ is the matrix formed from~$\calM$ by deleting these loops.
If $q=0$ then $\ZtildeTutte(\calM';q,\gamma)=1$.
Otherwise, $q<0$.
Now Lemma~\ref{lem:matroid} shows that
$\Ztilde(\calM';q,\gamma)> 0$ and can be computed in $\numPQ$. \end{proof} 

The following corollary follows immediately using Equation~(\ref{eq:graphmatroid}).

\begin{corollary}\label{cor:K}
Let $(x,y)$ be a point with $x\geq 1$ and $-1\leq y < 0$. 
Let $q=(x-1)(y-1)$ and $\gamma=y-1$.
Then   $\signtutte(q,\gamma)$ is in $\FP$ and $\tutte(q,\gamma)$ is in $\numPQ$.
\end{corollary}

\subsubsection{Points in Region J}
\label{sec:posJ}
 
The following lemma classifies points in Region~J. See also \cite[Theorem 4.4]{JacksonSokal}. 
 
\begin{lemma}
Let $(x,y)$ be a point with $-1\leq x \leq 0$ and $y\geq 1$.
Let $q=(x-1)(y-1)$ and $\gamma=y-1$.
Then  $\msigntutte(q,\gamma)$ is in $\FP$ and $\mtutte(q,\gamma)$ is in $\numPQ$.
\end{lemma}

\begin{proof} This follows from Fact~\ref{fact:duality} and from Lemma~\ref{lem:K}. \end{proof} 

The following corollary follows immediately using Equation~(\ref{eq:graphmatroid}).
\begin{corollary}\label{cor:J}
Let $(x,y)$ be a point with $-1\leq x \leq 0$ and $y\geq 1$.
Let $q=(x-1)(y-1)$ and $\gamma=y-1$.
Then  $\signtutte(q,\gamma)$ is in $\FP$ and $\tutte(q,\gamma)$ is in $\numPQ$.
\end{corollary}

\subsection{Points in Regions L and M}
We use the following Lemma.
The statement is a slight generalisation of \cite[Theorem 5.4]{JacksonSokal}.
However, their proof (a straightforward generalisation of their
proof of   \cite[Theorem 5.1]{JacksonSokal}) suffices.

\begin{lemma} (Jackson and Sokal)
\label{newJS}
Let $\calM$ be a matroid with ground set~$E$ and let $q\in(0,1)$.
 Suppose that $\bgamma$ is a weight function such that
\begin{enumerate}
\item $\gamma_e> -1$ for every loop~$e$;
\item $\gamma_e < -q$ for every coloop~$e$; and
\item $-1-\sqrt{1-q}< \gamma_e < -1+\sqrt{1-q}$ for every normal (i.e., non-loop and
non-coloop) element~$e$
\end{enumerate}
Then 
\begin{equation}
\label{newsigneq} {(-1)}^{r_\calM(E)} \ZtildeTutte(\calM;q,\bgamma)>0
\end{equation}
 and
the problem of computing 
$\ZtildeTutte(\calM;q,\bgamma)$, given such a matroid $\calM$ is in $\numPQ$.
\end{lemma}

\begin{proof} We follow the inductive argument that Jackson and Sokal use to prove (\ref{newsigneq}) for the graphical case.
This is the proof of \cite[Theorem 5.1]{JacksonSokal}.
The induction is on~$m$,  the  number of elements  in the ground set of~$\calM$.
If  $m=0$, then $r_\calM(E) = 0$ so 
$\ZtildeTutte(\calM;q,\bgamma) = 1$, so the lemma is true.
For $m>0$, there are five cases. We apply these in order, so in each case we assume that the previous cases don't apply.
\begin{enumerate}
\item If $\calM$ has a loop~$e$ then by Fact~\ref{fact:loop}, 
$$\Ztilde(\calM;q,\bgamma) = (1+\gamma_e) \Ztilde(\calM \setminus e;q,\bgamma).$$
Note that $1+\gamma_e > 0$ 
and $r_{\calM \setminus e}(E\setminus e) = r_\calM(E\setminus e) = r_\calM(E)$. 
Thus, the result follows by induction.
\item If $\calM$ has a coloop~$e$ then by Fact~\ref{fact:bridge},
$$\Ztilde(\calM;q,\bgamma) = (1+\gamma_e/q) \Ztilde(\calM \contract e;q,\bgamma).$$
Note that $1+\gamma_e/q<0$ and $r_{\calM\contract e}(E\setminus e) = 
r_\calM(E)-r_\calM(e)=
r_\calM(E)-1$.
Thus, the result follows by induction.
\item Suppose that $\calM$ has a size-$2$ circuit consisting of edges~$e_1$ and~$e_2$.
Let $\calM'$ be the matroid formed from~$\calM$ by deleting~$e_2$
and let $\bgamma'$ be the weight function that is the same as $\bgamma$ 
except that $\gamma'_{e_1}$ is the effective weight from the parallel composition of~$e_1$
and~$e_2$ --- $\gamma'_{e_1} = \gamma_{e_1} + \gamma_{e_2} + \gamma_{e_1} \gamma_{e_2}$.
Then, as in the proof of Lemma~\ref{lem:matroid} (see 
\cite[(2.34)]{JacksonSokal}),
$\Ztilde(\calM;q,\bgamma) = \Ztilde(\calM';q,\bgamma')$. Also,
 $  r_{\calM'}(E\setminus e_2)
 = r_\calM(E \setminus e_2) = r_\calM(E)$. Finally,
Jackson and Sokal show that $\calM'$ and $\bgamma'$ satisfy the conditions of the lemma 
(so $\Ztilde(\calM;q,\bgamma)$ can be computed by induction).
\item Suppose that $\calM$ has a size-$2$ cocircuit  
 consisting of edges~$e_1$ and~$e_2$.
Let $\calM'$ be the matroid formed from~$\calM$ by contracting~$e_2$
and let $\bgamma'$ be the weight function that is the same as $\bgamma$ 
except that $\gamma'_{e_1}$ is the effective weight from the series composition of~$e_1$
and~$e_2$ --- $\gamma'_{e_1} =  \gamma_{e_1} \gamma_{e_2} / (q+\gamma_{e_1} + \gamma_{e_2})$.
Then from \cite[(2.40)]{JacksonSokal}
$\Ztilde(\calM;q,\bgamma) = \left(\frac{q+\gamma_{e_1} + \gamma_{e_2}}{q}\right)\Ztilde(\calM';q,\bgamma')$. Also,
Jackson and Sokal show that
$$ \left(\frac{q+\gamma_{e_1} + \gamma_{e_2}}{q}\right)<0.$$
This is what we require, since
 $  r_{\calM'}(E\setminus e_2) = r_\calM(E)-r_\calM(e_2) = r_\calM(E)-1$.
Finally,
Jackson and Sokal show that $\calM'$ and $\bgamma'$ satisfy the conditions of the lemma 
(so $\Ztilde(\calM;q,\bgamma)$ can be computed by induction). 
\item Otherwise, pick any ground set element~$e$ and
apply the deletion-contraction identity \cite[(2.29a)]{JacksonSokal}
$$\Ztilde(\calM;q,\bgamma) = \Ztilde(\calM\setminus e;q,\bgamma) + \frac{\gamma_e}{q}\Ztilde(\calM / e;q, \bgamma).$$
Since $e$ is not a cocircuit, $r_{\calM \setminus e}(E\setminus e) = r_\calM(E)$.
As Jackson, and Sokal argue,   $\calM\setminus e$ and $\bgamma$ satisfy the conditions of the lemma.
Also, $\gamma_e/q<0$
and $r_{\calM/e}(M\setminus e) = r_{\calM}(E)-1$. Again, Jackson and Sokal argue that
  $\calM/e$ and $\bgamma$ satisfy the conditions of the lemma, so the result follows by induction.
\end{enumerate}  
\end{proof} 
 
 \subsection{Points in Region L}
  
\begin{lemma} 
\label{lem:L}
Let $(x,y)$ be a point with $0<x<1$ and $-x<y<0$.
Let $q=(x-1)(y-1)$ and $\gamma=y-1$. Then   $\msigntutte(q,\gamma)$ is in $\FP$ and $\mtutte(q,\gamma)$ is in $\numPQ$.
\end{lemma}
\begin{proof}
Note that $q = (1-x)(1-y) < (1-x)(1+x) = 1-x^2 < 1$.
Also, $q > (1-x) > 0$. Thus, $q\in (0,1)$.

Now since $y>-x$ we have $y(y-1) < (-x)(y-1)$ 
so $y^2-y< x-xy$
which implies $y^2 < x+y-xy=1-q$.
This implies that $y < |y|< \sqrt{1-q}$ 
so $y > -\sqrt{1-q}$. Thus, 
$-1-\sqrt{1-q} < \gamma <-1 + \sqrt{1-q}$.
 
Finally, since $0<x(1-y)$, 
we have $y<y+x(1-y) = 1- q$ so
$\gamma < -q$.

Now let $\calM$ be a matroid and
let $\bgamma$ be a weight function assigning weight~$\gamma$ to every element the ground-set  of $\calM$.
If $\calM$ has $k$  loops then  by Fact~\ref{fact:loop}, 
$\ZtildeTutte(\calM;q,\bgamma) ={ (1+\gamma)}^k \ZtildeTutte(\calM';q,\bgamma)$,
where $\calM'$ is the matroid formed from~$\calM$ by deleting these loops.
Note that $\calM'$ and $\bgamma$ satisfy the hypotheses of Lemma~\ref{newJS}.   
\end{proof}

The following corollary follows immediately using Equation~(\ref{eq:graphmatroid}).
\begin{corollary}\label{cor:L}
Let $(x,y)$ be a point with $0<x<1$ and $-x<y<0$.
Let $q=(x-1)(y-1)$ and $\gamma=y-1$. Then   $\signtutte(q,\gamma)$ is in $\FP$ and $\tutte(q,\gamma)$ is in $\numPQ$.
\end{corollary}

\subsection{Points in Region M}

\begin{lemma} 
\label{lem:M}
Let $(x,y)$ be a point with $0<y<1$ and $-y<x<0$.
Let $q=(x-1)(y-1)$ and $\gamma=y-1$. Then   $\msigntutte(q,\gamma)$ is in $\FP$ 
and $\mtutte(q,\gamma)$ is in $\numPQ$.
\end{lemma}
 \begin{proof}
This follows from Fact~\ref{fact:duality} and from Lemma~\ref{lem:L}.
\end{proof}

The following corollary follows immediately using Equation~(\ref{eq:graphmatroid}).
 
\begin{corollary}\label{cor:M} 
Let $(x,y)$ be a point with $0<y<1$ and $-y<x<0$.
Let $q=(x-1)(y-1)$ and $\gamma=y-1$. Then   $\signtutte(q,\gamma)$ is in $\FP$ 
and $\tutte(q,\gamma)$ is in $\numPQ$.
\end{corollary} 
 
\section{Putting things together for points with $|y|<1$}

Collecting Observations~\ref{obs:posBF} and~\ref{obs:Fpos} and Corollaries~\ref{cor:B},
\ref{cor:BGverticalboundary}, \ref{cor:GIH},  \ref{cor:F},  \ref{cor:FBboundary} and \ref{cor:GFboundary}.

\begin{corollary}
Suppose $(x,y)$ is a point satisfying $|y|<1$ such that $q=(x-1)(y-1)\geq 32/27$.
Let $\gamma=y-1$.
\begin{itemize}
\item  If $(x,y)=(-1,0)$ then $\signtutte(q,\gamma)$ and $\tutte(q,\gamma)$ are in $\FP$.
\item  If $(x,y)=(x,0)$ for any integer $x<-1$ then $\signtutte(q,\gamma)$ is NP-complete.
$\tutte(q,\gamma)$ is in  $\numPQ$.
\item If $x\leq -1$ and $0 < y < 1$ and $q$ is an integer
then $Z(G;q,\gamma)> 0$ so $\signtutte(q,\gamma)$ is in $\FP$.
Also, $\tutte(q,\gamma)$ is in $\numPQ$.
\item Otherwise,  $\signtutte(q,\gamma)$ is $\numP$-hard.
\end{itemize}
\label{cor:meetJS}
\end{corollary}

\section*{Acknowledgement}
The authors are grateful to Bill Jackson for 
pointing out that computing the sign is NP-hard
at the point $(0,-3)$.  
  
\bibliographystyle{plain}
\bibliography{mybibfile}
 
\end{document}